\documentclass[letterpaper,11pt]{article}
\usepackage[in]{fullpage}
\usepackage[normalem]{ulem}
\usepackage{amsmath}
\usepackage{amsthm}
\usepackage{amssymb}
\usepackage{amsfonts}
\usepackage{xcolor}
\usepackage{soul}
\usepackage{xfrac}
\usepackage{xspace}
\usepackage{xstring}
\usepackage{comment}
\usepackage{embedfile}
\usepackage{algorithm}
\usepackage[noend]{algpseudocode}
\usepackage{thmtools}
\usepackage{parskip}
\usepackage{hyperref}
\hypersetup{
	colorlinks,
	citecolor=blue,
	filecolor=black,
	linkcolor=black,
	urlcolor=black,
	linktoc=all
}

\newif\iftoc

\newif\iflineno

\iflineno
\usepackage{lineno}
\linenumbers
\fi

\newtheorem{lemma}{Lemma}[section]
\newtheorem{theorem}[lemma]{Theorem}
\newtheorem{corollary}[lemma]{Corollary}
\newtheorem{observation}[lemma]{Observation}

\newtheorem*{theorem*}{Observation}
\newtheorem*{lemma*}{Lemma}
\newtheorem*{observation*}{Observation}

\theoremstyle{definition}
\newtheorem{definition}[lemma]{Definition}

\newcommand{\accept}[0]{\textsc{accept}\xspace}
\newcommand{\reject}[0]{\textsc{reject}\xspace}

\newcommand{\ForRange}[3]{\For{{#1} \textbf{from} {#2} \textbf{to} {#3}}}

\newcommand{\braket}[1]{\left\langle #1 \right\rangle}
\newcommand{\card}[1]{\left|{#1}\right|}
\newcommand{\cond}{\middle |}

\newcommand{\ceil}[1]{{\left\lceil{#1}\right\rceil}}
\newcommand{\lfrac}[2]{\left. {#1} \middle/ {#2} \right.}

\newcommand{\poly}{\mathrm{poly}}
\newcommand{\supp}{\mathrm{supp}}

\newcommand{\set}[1]{{\left\{#1\right\}}}

\newcommand{\eps}[0]{\varepsilon}
\newcommand{\E}{\mathop{{\rm E}\/}}

\title{Refining the Adaptivity Notion in the Huge Object Model}

\author{Tomer Adar\thanks{Technion - Israel Institute of Technology, Israel. Email: \href{mailto:tomer-adar@campus.technion.ac.il}{tomer-adar@campus.technion.ac.il}.} \and Eldar Fischer\thanks{Technion - Israel Institute of Technology, Israel. Email: \href{mailto:eldar@cs.technion.ac.il}{eldar@cs.technion.ac.il}. Research supported by an Israel Science Foundation grant number 879/22.}}

\begin{document}
	\begin{titlepage}
		\maketitle
		\thispagestyle{empty}
		\begin{abstract}
			The Huge Object model for distribution testing, first defined by Goldreich and Ron in 2022, combines the features of classical string testing and distribution testing. In this model we are given access to independent samples from an unknown distribution $P$ over the set of strings $\{0,1\}^n$, but are only allowed to query a few bits from the samples. The distinction between adaptive and non-adaptive algorithms, which occurs naturally in the realm of string testing (while being irrelevant for classical distribution testing), plays a substantial role also in the Huge Object model.
			
			In this work we show that the full picture in the Huge Object model is much richer than just that of the ``adaptive vs.\ non-adaptive'' dichotomy. We define and investigate several models of adaptivity that lie between the fully-adaptive and the completely non-adaptive extremes. These models are naturally grounded by observing the querying process from each sample independently, and considering the ``algorithmic flow'' between them. For example, if we allow no information at all to cross over between samples (up to the final decision), then we obtain the {\em locally bounded} adaptive model, arguably the ``least adaptive'' one apart from being completely non-adaptive. A slightly stronger model allows only a ``one-way'' information flow. Even stronger (but still far from being fully adaptive) models follow by taking inspiration from the setting of streaming algorithms. To show that we indeed have a hierarchy, we prove a chain of exponential separations encompassing most of the models that we define.
		\end{abstract}
	\end{titlepage}
	
	\iftoc
	\tableofcontents
	\thispagestyle{empty}
	\newpage
	\pagenumbering{arabic}
	\fi
	
	\section{Introduction}
	
	Property testing is the study of sublinear, query-based probabilistic decision-making algorithms. That is, algorithms that return \accept or \reject after reading only a small portion of their input. The study of (classical) property testing, starting with \cite{blr1993}, \cite{rs1992} and \cite{rs1996}, has seen an extensive body of work. See for example 
	\cite{g17}. Usually, a property-testing algorithm with threshold parameter $\eps$ is required to accept an input that satisfies the property with high probability, and reject an input whose distance from any satisfying one is more than $\eps$, with high probability as well. For string properties, which were the first to be studied (along with functions, matrices, etc.\ that can also be represented as strings), the distance measure is usually the normalized Hamming distance.
	
	Distribution testing is a newer model, first defined implicitly in \cite{gr2011} (a version of which has already appeared in 2000 as a technical report). In \cite{batuFFKRW2001} and \cite{batu-FRSW2000} it was explicitly defined and researched. The algorithms in this model are much weaker, where instead of queries, the decision to accept or reject must be made based only on a sequence of independent samples drawn from an unknown distribution. In such a setting the distance metric is usually the variation distance. For a more comprehensive survey, see \cite{gs009}.
	
	The study of a \emph{combination} of string and distribution testing was initiated in \cite{gr2022}. Here the samples in themselves are considered to be very large objects, and hence after obtaining a sample (usually modeled as a string of size $n$), queries must be made to obtain some information about its contents. This requires an appropriate modification in the distance notion. This model is appropriately called the \emph{Huge Object model}.
	
	Contrast the above to the original ``small object'' distribution testing model, where it is assumed that every sample is immediately available to the algorithm in its entirety. In particular, in the original model, the algorithm does not have any choice of queries, as it just receives a sequence of independent samples from the distribution to be tested. Hence one might even call it a ``formula'' rather than an ``algorithm''. Grossly speaking, the only decision made is whether to accept or reject the provided sequence of sampled objects. 
	
	On the other hand, in the string testing model, an algorithm is provided with a (deterministic) input string, and may make query decisions based both on internal random coins and on answers to previous queries. An algorithm which makes use of the option of considering answers to previous queries when choosing the next query is called adaptive, while an algorithm that queries based only on coin tosses is called non-adaptive (the final decision on whether to accept or reject the input must, of course, depend on the actual answers).
	
	Algorithms for the Huge Object model, due to their reliance on individual queries to the provided samples, can be adaptive or non-adaptive. This relationship with respect to the Huge Object model was first explored in \cite{fischer2022}.
	
	However, as we shall demonstrate below, the complete picture here is richer than the standard adaptive/non-adaptive dichotomy used in classical string testing. As it turns out, several categories of adaptivity can be defined and investigated based on the consideration of the shared information between the different samples that are queried.
	
	\subsection{Adaptivity notions in the Huge Object model}
	
	For our purpose, unless we state otherwise, we assume that the sequence of samples is taken in advance (but is not directly disclosed to the algorithm), and is presented as a matrix from which the algorithm makes its queries. For a sequence of $s$ samples from a distribution whose base set is $\{0,1\}^n$, this would be a binary $s\times n$ matrix.
	
	We say that an algorithm is \emph{non-adaptive} if it chooses its entire set of queries before making them, which means that it cannot choose later queries based on the answers to earlier ones. This is identical to the definition of a non-adaptive algorithm for string properties.
	
	A \emph{fully adaptive} algorithm is allowed to choose every query based on answers to all queries made before it. This is quite similar to the definition of an adaptive algorithm for string properties, but restricting ourselves to this dichotomy does not give the full picture. We refine the notion of adaptivity by considering more subtle restrictions on the way that the algorithms plan their queries, leading to query models that are not as expressive as those of fully adaptive algorithms, but are still more expressive than those of non-adaptive ones. In this introduction we only introduce the rationale of every model; the formal definitions appear in the preliminaries section.
	
	One interesting restriction, which is surprisingly difficult to analyze, is ``being adaptive for every individual sample, without sharing adaptivity between different samples'' (the results of random coin tosses are still allowed to be shared). We say that an algorithm is \emph{locally-bounded} if it obeys this restriction. This model captures the concept of distributed execution, in a way that every node has a limited scope of a single sample, and only when all nodes are done, their individual outcomes are combined to facilitate a decision.
	
	A more natural restriction is ``being able to query only the most recent sample''. We say that an algorithm is \emph{forward-only} if it cannot query a sample after querying a later one. This can be viewed (if we abandon the above-mentioned ``matrix representation'') as the algorithm being provided with oracle access to only one sample at a time, not being able to ``go back in time'' once a new sample was taken. An example for the usage of the model is an anonymous survey. As long as the survey session is alive, we can present new questions based on past interactions and on the current one, but once the session ends, we are not able to recall the same participant for further questioning.
	
	A natural generalization of forward-only adaptiveness is having a bounded memory for holding samples (rather than only having one accessible sample at a time). Once the memory is full, the algorithm must drop one of these samples (making it inaccessible) in order to free up space for a new sample. An additional motivation for this model is the concept of stream processing, whose goal is computing using sublinear memory. Relevant to our work is \cite{aliakbarpour2022estimation}, where the input stream is determined by an unknown distribution, in contrast to the usual streaming setting where the order of the stream is arbitrary. Within the notion of having memory of a fixed size, we actually distinguish two models. In the weak model, when the memory is full, the oldest sample is dropped. In the strong model, the algorithm decides (possibly adaptively) which sample to drop.
	
	We show that every two consecutive models in the above hierarchy have an exponential separation, which means that there is a property that requires $\Omega(\poly(n))$ queries for an $\eps$-test in the first model (for some fixed $\eps$), but is also $\eps$-testable using $O\left(\poly\left(\eps^{-1}\right) \log n\right)$ queries in the second model (for every $\eps > 0$). Moreover, our upper bounds always have one-sided error, while the lower bounds apply for both one-sided and two-sided error algorithms. The exact relationship between the weak and the strong limited memory models remains open, however.
	
	We believe that investigating limited adaptiveness models can apply to other areas where there are two ``query scales''. That is, when investigating a model takes into account collections of objects that are restricted both in the way that whole objects are obtained and in the access model \emph{inside} each obtained object. For example, one could think of a distributed computing scenario where the communication between the nodes follows a LOCAL or a CONGEST scheme (see \cite{peleg2000distributed}), but additionally each node holds a ``large'' input from which it may only perform sub-linear time computation \emph{between} the communication rounds.
	
	\subsection{Organization of the paper}
	
	We start with formal definitions of the models which are required to state our results, followed by an overview of the results themselves and a description of the main ideas of their proofs. The overview also serves as a guide to the rest of the paper, that contains the formal proofs. While the statements in the overview are labeled as ``informal'', the main difference between them and the formal statements to which they refer is that the latter also specify the specific properties that demonstrate the query bounds.
	
	
	\section{Foundational preliminaries}
	
	The following are the core definitions and lemmas used throughout this paper, including the model definitions used in the overview in Section \ref{sec:overview}. Here, all distributions are defined over finite sets.
	
	\begin{definition}[Common notations]
		For a set $A$, \emph{the power set of $A$} is denoted by $\mathcal{P}(A)$. For two sets $A$ and $B$, the set of all functions $f : A \to B$ is denoted by $B^A$. For a finite set $A$, the \emph{set of all permutations over $A$} is denoted by $\pi(A)$.
	\end{definition}
	
	%
	%
	
	\begin{definition}[Set of distributions]
		Let $\Omega$ be a finite set. The \emph{set of all distributions that are defined over $\Omega$} is denoted by $\mathcal{D}(\Omega)$.
	\end{definition}
	While parts of this section are generalizable to distributions over non-finite sets $\Omega$ with compact topologies, we restrict ourselves to distributions over finite sets, which suffice for our application.
	
	\begin{definition}[Property]
		A \emph{property} $\mathcal{P}$ over a finite alphabet $\Sigma$ is defined as a sequence of compact sets $\mathcal{P}_n\subseteq\mathcal{D}(\Sigma^n)$. Here {\em compactness} refers to the one defined with respect to the natural topology inherited from $\mathbb{R}^{|\Sigma|^n}$.
	\end{definition}
	All properties are defined over $\Sigma = \{0,1\}$ unless we state otherwise.
	
	\subsection{Distances}
	
	The following are the distance measures that we use. In the sequel, we will omit the subscript (e.g.\ use ``$d(x,y)$'' instead of ``$d_\mathrm{H}(x,y)$'') whenever the measure that we use is clear from the context.
	
	\begin{definition}[Normalized Hamming distance]
		For two strings $s_1, s_2 \in \Sigma^n$, we use $d_\mathrm{H}(s_1,s_2)$ to denote their \emph{normalized Hamming distance}, $\frac{1}{n} \left| \left\{ 1 \le i \le n \middle| s_1[i] \ne s_2[i] \right\} \right|$.
	\end{definition}
	
	For all our distance measures we also use the standard extension to distances between sets, using the corresponding infimum (which in all our relevant cases will be a minimum). For example, For a string $s \in \{0,1\}^n$ and a set $A \subseteq \{0,1\}^n$, we define $d_\mathrm{H}(s,A) = \min\limits_{s' \in A} d_\mathrm{H}(s,s')$.
	
	\begin{definition}[Variation distance]
		For two distributions $P$ and $Q$ over a common set $\Omega$, we use $d_\mathrm{var}(P,Q)$ to denote their \emph{variation distance}, $\max_{E \subseteq \Omega} {\left| \Pr_P[E] - \Pr_Q[E] \right|}$. Since $\Omega$ is finite there is an equivalent definition of $d_\mathrm{var}(P,Q) = \frac{1}{2}\sum_{s \in \Omega} {\card{P(s)-Q(s)}}$.
	\end{definition}
	
	\begin{definition}[Transfer distribution]
		For two distributions $P$ over $\Omega_1$ and $Q$ over $\Omega_2$, we say that a distribution $T$ over $\Omega_1 \times \Omega_2$ is a {\em transfer distribution} between $P$ and $Q$ if for every $x_0 \in \Omega_1$, $\Pr_{(x,y) \sim T}[x = x_0] = \Pr_{P}[x_0]$, and for every $y_0 \in \Omega_2$, $\Pr_{(x,y) \sim T}[y = y_0] = \Pr_{Q}[y_0]$. We use $\mathcal{T}(P,Q)$ to denote the set of all transfer distributions between $P$ and $Q$.
	\end{definition}
	
	We note that for finite $\Omega_1$ and $\Omega_2$ the set $\mathcal{T}(P,Q)$ is compact as a subset of $\mathcal{D}(\Omega_1\times\Omega_2)$.
	
	\begin{definition}[Earth Mover's Distance]
		For two distributions $P$ and $Q$ over a common set $\Omega$ with a metric $d_\Omega$, we use $d_\mathrm{EMD}(P,Q)$ to denote their \emph{earth mover's distance}, defined by the infimum of the ``average distance'' demonstrated by a transfer distribution,
		$\inf_{T \in \mathcal{T}(P,Q)}{ \E_{(x,y) \sim T}\left[d_{\Omega}(x,y)\right]}$.
	\end{definition}
	
	In the sequel, the above ``$\inf$'' can and will be replaced by ``$\min$'', by the compactness of $\mathcal{T}(P,Q)$ for finite $\Omega$. Most papers (including the original \cite{gr2022}) use an equivalent definition that is based on linear programming, whose solution is the optimal transfer distribution.
	
	In our theorems, $\Omega$ is always ${\{0,1\}}^n$ for some $n$ and the metric is the Hamming distance. Sometimes, as an intermediate phase, we may use a different $\Omega$ (usually $\{1,\ldots,k\}^n$ for some $k$), and then show a reduction back to the binary case.
	
	\begin{definition}[Distance from a property] \label{def:dist-from-property}
		The \emph{distance} of a distribution $P$ from a property $\mathcal{P} = \langle \mathcal{P}_n \rangle$ is loosely noted as $d(P,\mathcal{P})$ and is defined to be $d_{\mathrm{EMD}}(P,\mathcal{P}_n) = \inf_{Q \in \mathcal{P}_n} d_\mathrm{EMD}(P,Q)$.
	\end{definition}
	
	It is very easy to show that for any two distributions $P,Q\in\mathcal{D}(\Sigma^n)$ we have $d_{\mathrm{EMD}}(P,Q)\leq d_{\mathrm{var}}(P,Q)$. This means that the topology induced by the variation distance is richer than that induced by the earth mover's distance (actually for finite sets these two topologies are identical). In particular it means that all considered properties form compact sets with respect to the earth mover's distance. We obtain the following lemma.
	
	\begin{lemma}\label{lemma:emd-supp}
		For a property $\mathcal{P}$ of distributions over strings, and any distribution $P\in\mathcal{D}(\Sigma^n)$, there is a distribution realizing the distance of $P$ from $\mathcal{P}$, i.e.\ a distribution $Q\in\mathcal{P}_n$ for which $d(P,Q)=d(P,\mathcal{P}_n)$. In particular, the infimum in Definition \ref{def:dist-from-property} is a minimum.
	\end{lemma}
	
	\subsection{The testing model}
	
	This model is defined in \cite{gr2022}. We use an equivalent definition which will be the ``baseline'' for our restricted adaptivity variants.
	
	The input is a distribution $P$ over $\Sigma^n$ (our final theorems will be for $\Sigma=\{0,1\}$, but some lemmas will have other finite $\Sigma$). An algorithm $\mathcal{A}$ gets random oracle access to $s$ samples that are independently drawn from $P$. Then it is allowed to query individual bits of the samples. The output of the algorithm is either \accept or \reject. For convenience we identify the samples with an $s\times n$ matrix, so for example the query ``$(i,j)$'' returns the $j$th bit of the $i$th sample.
	
	The input size $n$ and the number of samples $s$ are hard-coded in the algorithm. As with boolean circuits, an algorithm for an arbitrarily sized input is defined as a sequence of algorithms, one for each $n$.
	
	For a given algorithm we define another measure of complexity, which is the total number of queries that the algorithm makes. Without loss of generality, we always assume that every sample is queried at least once (implying that $q \ge s$).
	
	For a property $\mathcal{P}$ and $\eps > 0$, we say that an algorithm $\mathcal{A}$ is an $\eps$-test if:
	\begin{itemize}
		\item For every $P \in \mathcal{P}$, $\mathcal{A}$ accepts the input $P$ with probability higher than $\frac{2}{3}$.
		\item For every $P$ that is $\eps$-far from $\mathcal{P}$, $\mathcal{A}$ accepts the input $P$ with probability less than $\frac{1}{3}$.
	\end{itemize}
	We say that $\mathcal{A}$ is an $\eps$-test with one sided error if:
	\begin{itemize}
		\item For every $P \in \mathcal{P}$, $\mathcal{A}$ accepts the input $P$ with probability 1.
		\item For every $P$ that is $\eps$-far from $\mathcal{P}$, $\mathcal{A}$ accepts the input $P$ with probability less than $\frac{1}{2}$.
	\end{itemize}
	
	The choice of the probability bounds in the above definition are somewhat arbitrary. For the one sided error definition $\frac{1}{2}$ is more convenient than $\frac{1}{3}$. We also note that for non-$\eps$-far inputs that are not in $\mathcal{P}$, any answer by $\mathcal{A}$ is considered to be correct.
	
	\subsection{Restricted models}\label{sec:prelim:subsec:models}
	
	As observed by Yao in \cite{yao1977}, every probabilistic algorithm can be seen as a distribution over the set of allowable deterministic algorithms. This simplifies the algorithmic analysis, since we only have to consider deterministic algorithms (a distinction between public and private coins may break this picture, but this will not be the case here). We will use Yao's observation to define every probabilistic algorithmic model by defining its respective set of allowable deterministic algorithms.
	
	\begin{definition}[Fully adaptive algorithm] \label{def:det-alg-ho-model}
		Every deterministic algorithm can be described as a full decision tree $T$ and a set $A$ of accepted leaves. Without loss of generality we assume that all leaves have the exactly the same depth (we use dummy queries if ``padding'' is needed). Every internal node of $T$ consists of a query $(i,j) \in \{1,\ldots,s\} \times \{1,\ldots,n\}$ (the $j$th bit of the $i$th sample), and every edge corresponds to an outcome element (in $\Sigma$). The number of queries $q$ is defined as the height of the tree. Every leaf can be described by the string of length $q$ detailing the answers given to the $q$ queries, corresponding to its root-to-leaf path. Thus we can also identify $A$ with a subset of $\Sigma^q$. We use variants of the decision tree model to describe our adaptivity concepts.
	\end{definition}
	
	Now that we have defined the most general form of a deterministic algorithm in the Huge Object model, we formally define our models for varying degrees of adaptivity.
	
	\begin{definition}[Non-adaptive algorithm]
		We say that an algorithm is \emph{non-adaptive} if it chooses its queries in advance, rather than deciding each query location based on the answers to its previous ones. Formally, every deterministic non-adaptive algorithm is described as a pair $(Q,A)$ such that $Q \subseteq \{1,\ldots,s\} \times \{1,\ldots,n\}$ (for some sample complexity $s$) is the set of queries, and $A \subseteq \Sigma^Q$ is the set of accepted answer functions. The query complexity is defined as $q = \card{Q}$.
	\end{definition}
	
	\begin{definition}[Locally-bounded adaptive algorithm]\label{def:local-alg}
		We call an algorithm \emph{locally-bounded} if it does not choose its queries to a sample based on answers to queries in other samples. Formally, every $s$-sample deterministic locally-bounded algorithm is a tuple $\left(T_1,\ldots,T_s;A\right)$, where every $T_i$ is a decision tree of height $q_i$ (where $q = \sum_{i=1}^s q_i$ is the total number of queries) that is only allowed to query the $i$th sample, and $A \subseteq \Sigma^q$ represents a set of accepted superleaves, where a superleaf is defined as the concatenation of the $q_1,\ldots,q_s$ symbol long sequences that represent the leaves of trees $T_1,\ldots,T_s$ respectively.
	\end{definition}
	
	\begin{definition}[Forward-only adaptive algorithm]
		We call an algorithm \emph{forward-only} if it cannot query a sample after querying a later one. Formally, a forward-only algorithm for $s$ samples of $n$-length strings is defined as a pair $(T,A)$, where $T$ is a decision tree over $\{1,\ldots,s\} \times \{1,\ldots,n\}$ and $A \subseteq \Sigma^q$ (as with general adaptive algorithms), additionally satisfying that for every internal node of $T$ that is not the root, if its query is $(i,j)$ and its parent query is $(i',j')$, then $i' \le i$.
	\end{definition}
	
	\begin{definition}[Weak memory-bounded adaptive algorithm] \label{def:model:weak-k-mem}
		We say that an algorithm is \emph{weak $m$-memory bounded} if it can only query a sliding window of the $m$ most recent samples at a time. Formally, a weak $m$-memory-bounded adaptive algorithm using $s$ samples of $n$-length strings is defined as a pair $(T,A)$, where $T$ is a decision tree over $\{1,\ldots,s\} \times \{1,\ldots,n\}$ and $A \subseteq \Sigma^q$ (as with general adaptive algorithms), additionally satisfying that for every internal node of $T$ that is not the root, if its query is $(i,j)$, then for every ancestor whose query is $(i',j')$, it holds that $i'-m < i$.
	\end{definition}
	
	\begin{definition}[Strong memory-bounded adaptive algorithm] \label{def:model:strong-k-mem}
		A \emph{strong memory-bounded adaptive algorithm} for $s$ samples of $n$-length strings is defined as a triplet $(T,A,M)$ where $T$ is a decision tree, $A \subseteq \Sigma^q$ is the set of accepted answer vectors, and $M : \mathrm{nodes}(T) \to \mathcal{P}(\{1,\ldots,s\})$ is the ``memory state'' at every node. The explicit rules of $M$ are:
		\begin{itemize}
			\item For every internal node $u \in T$, $\card{M(u)} \le k$ (there are at most $k$ samples in memory).
			\item For every internal node $u \in T$, if $i \in M(u)$, and if $v$ is a child of $u$ for which $i \notin M(v)$, then for every descendant $w$ of $v$, $i \notin M(w)$ (a ``forgotten'' sample cannot be ``recalled'').
			\item For every internal node $u \in T$ whose query is $(i,j)$, $i \in M(u)$ (the $i$th sample must be in memory in order to query it).
		\end{itemize}
	\end{definition}
	Without loss of generality, because the samples are independent, we can assume that:
	\begin{itemize}
		\item $M(\mathrm{root}) = \{1,\ldots,k\}$ (the algorithm has initial access to the first $k$ samples).
		\item For every internal node $u \in T$ and the set $V$ of all its ancestors, it holds that $\max(M(u)) \le 1 + \max\limits_{v \in V} (\max M(v))$ (new samples are accessed ``in order'').
	\end{itemize}
	
	\section{Overview of results and methods}
	\label{sec:overview}
	
	The following is an informal overview of our work. Most of our results are exponential separations between models (that is, $O(\log n)$ vs $n^{\Omega(1)}$ bounds), but we also define new methodologies and analyze example properties.
	
	All separations are with an exponential gap, and are achieved by properties that have an efficient $1$-sided error test in one model, but do not even have an efficient $2$-sided test in the other model.
	
	
	Figure \ref{fig:summary-of-our-results} {provides} a visualization of our results. More details about the difference between the weak $k$-memory and the strong $k$-memory model are provided below.
	
	\begin{figure}
		\centering
			\includegraphics{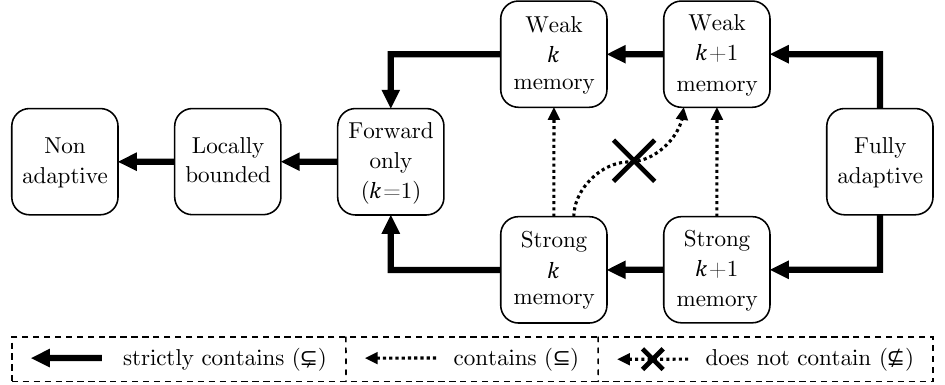}
		\caption{Graphical summary of our results}
		\label{fig:summary-of-our-results}
	\end{figure}
	
	
	\subsection{Non-adaptive algorithms}
	In section \ref{sec:na} we analyze three example properties to examine the similarities and differences between the Huge Object model (when restricted to non-adaptive queries) and the classic sampling model. Here we shortly describe two of them.
	
	We show that the determinism property (the property of all distributions that draw a specific element with probability 1) can be tested non-adaptively using $O(\eps^{-1})$ queries, consisting of $O(\eps^{-1})$ samples (as in the classic model) and $O(1)$ queries per sample.
	\begin{observation*}(Informal statement of Observation \ref{obs:na:determinism} regarding Algorithm \ref{alg:na:determinism})
		The property of drawing a fixed string has a one-sided error non-adaptive $\eps$-test that uses $O(\eps^{-1})$ queries.
	\end{observation*}
	
	The immediate generalization of the determinism property is the bounded support property.
	\begin{theorem*}(Informal statement of Theorem \ref{th:na:m-support} about Algorithm \ref{alg:na:m-support})
		The property of being supported on a set of at most $m$ elements has a one-sided error non-adaptive $\eps$-test that uses $O(\eps^{-2} m \log m)$ queries.
	\end{theorem*}
	As described in detail in Section \ref{sec:na}, our $\eps$-test for the $m$-support property needs more than a fixed number of queries per sample. Though not necessarily optimal, this algorithm demonstrates the core difference between the Huge Object model and the classic one: the limited ability to distinguish different samples. This limitation holds for adaptive algorithms as well, even though the adaptivity can reduce the number of queries per sample for some properties.
	
	\subsection*{Locally bounded adaptive algorithms}
	The locally-bounded adaptive model allows the algorithm to pick its queries based on answers to previous queries for every \emph{fixed} sample, but lacks the ability to pass information between samples. The ability of being adaptive allows the algorithm more ways to query its samples, but it still lacks the ability to test \emph{relations} between the samples.
	
	\paragraph{Analysis method} To analyze the locally-bounded model, we define an intermediate model of string testing which we call the \emph{split-adaptive model}. In this model, we test properties of $k$-tuples of strings, where the queries are made separately for every entry of the tuple (that is, every entry is processed using an adaptive algorithm that is oblivious of the other entries). To obtain a reduction, we consider every $s$-sample locally-bounded algorithm over an input distribution $P$ as a split-adaptive algorithm whose input is drawn from $P^s$ (that is, an $s$-tuple whose entries are independently drawn from $P$).
	
	\paragraph{Exponential separation from the non-adaptive model}
	Naturally, there is an exponential separation between the locally-bounded model and the non-adaptive model of the Huge Object model.
	\begin{lemma}
		There exists a property $\mathcal{P}$ of distributions over $\{0,1\}^n$ that has a locally-bounded $\eps$-test that uses $O(\poly(\eps^{-1}) \log n)$ queries for every $\eps > 0$, but there exists some $\eps_0 > 0$ for which non-adaptive $\eps_0$-test requires $\Omega(\poly(n))$ queries.
	\end{lemma}
	
	This is an almost-direct corollary of a result from \cite{gr2022} regarding converting string testing problems to the Huge Object model. Essentially, the Huge Object model ``contains'' the string testing one, and the conversion produces locally adaptive algorithms out of their respective adaptive string algorithms.
	
	\subsection*{Forward only adaptive algorithms}
	In the forward-only model, the algorithm virtually gets a stream of samples, and is allowed to query only the current sample without any restriction (but further queries to past samples are not allowed), based on answers to all past queries. In contrast to the locally bounded model, forward algorithms can test a richer collection of binary relations between samples, due to the ability to query one sample and then use the gathered data to choose the queries for the next one.
	
	\paragraph{The query foresight method} Some adaptive algorithms do not obey the forward only restriction but can be modified to do so, using a method we call \emph{query foresight}. Intuitively, an adaptive algorithm that has some knowledge about the structure of the queries it may make in the future can make them speculatively at present (that is, we make all potential queries to satisfy the forward-only constraint, even though we believe that some of them will later be considered as irrelevant). The more knowledge the algorithm has about the potential future queries, the less queries are wasted on the current sample.
	
	As an example to the query foresight method, we present a fully adaptive algorithm for the $m$-support property (Algorithm \ref{alg:s-m+1-mem:m-supp}), which is usually better than the algorithm presented in the discussion about the non-adaptive algorithm. We observe that the general structure of its queries is highly predictable, and provide a modified version thereof (Algorithm \ref{alg:forward-m-supp}) which is also forward-only, without increasing its worst-case query complexity.
	
	\paragraph{Exponential separation from the locally-bounded model}
	We use the ability of forward-only algorithms to consider richer collection of relations between samples, as compared to locally-bounded algorithms, to show an exponential separation between these models.
	\begin{theorem*}(Informal, combined statement of Theorem \ref{th:lbnd-local-Inv} and Theorem \ref{th:alg-Inv-correct})
		There exists a property of distributions over $\{0,1\}^n$ that has a forward-only $\eps$-test that uses $O(\poly(\eps^{-1}) \log n)$ queries for every $\eps > 0$, but for which there exists some $\eps_0 > 0$ so that any locally-bounded adaptive $\eps_0$-test requires $\Omega(\poly(n))$ queries.
	\end{theorem*}
	
	In \cite{ergun99} it was shown that $\eps$-testing two functions over $\{1,\ldots,n\}$ for being inverses of each other is possible with $O(\eps^{-1})$ many queries, while testing a single function for having an inverse is harder and requires a polynomial number of queries. Here we separate the two functions by setting them in a probability space with support size $2$. If we allow forward-only adaptivity, then the original inverse test can be implemented, as it works by verifying that $g(f(i))=i$ for sufficiently many $i$s. We can call the first sample ``$f$'', and after writing down our $f(i_1),\ldots,f(i_q)$, we ``wait'' for a sample of $g$ and then verify that $g(f(i_j))=i_j$ for $i_1,\ldots,i_q$.
	
	To make the above work for binary strings (rather than an alphabet of size $n$) we use an appropriate large distance encoding of the values. Also, we modify the definition of the property slightly to make sure that it is possible to construct a one-sided test also using forward-only adaptivity.
	
	The lower bound against locally-bounded adaptivity requires an intricate analysis of the model. Essentially we use the split-adaptive string-testing model to show that when querying each of $f$ and $g$ ``in solitude'', being adaptive over a function that is drawn at random does not provide an advantage over a non-adaptive algorithm. In particular, the values of a uniformly drawn permutation are ``too random'' to allow the implementation of a meaningful query strategy without getting some information from the inverse function. Essentially, we show that ``coordinating in advance'' the query strategy is insufficient.
	
	\subsection*{$k$-bounded memory algorithms}
	As per Definitions \ref{def:model:weak-k-mem}, \ref{def:model:strong-k-mem} we have two kinds of bounded memory, which we call \emph{weak} and \emph{strong}. Intuitively, in both models, the algorithm gets a stream of samples, and it has an unrestricted access to $k$ of these samples. When the algorithm needs an access to a new sample, it must give up the ability to access one of the past samples. In the weak model, the algorithm does not have a choice and it must drop the earliest sample. In other words, the weak model has an unrestricted access to a sliding window of the $k$ most recent samples. In the strong model, the algorithm is allowed to choose the sample to drop.
	
	For $k=1$, the weak and strong models are both equal to each other and to the forward-only model. Intuitively, as $k$ increases, the algorithm is able to consider more complicated relations between samples, especially $k$-ary relations, which are more challenging for $k-1$-memory algorithms.
	
	\paragraph{Exponential separation from the forward-only model}
	We use the ability to fully consider binary relations using $2$-memory algorithms, compared to the slightly limited ability to do that using forward-only algorithms, to establish an exponential separation between them.
	
	\begin{theorem*}(Informal, combined statement of Theorem \ref{th:fwd-lbnd-sym} and Theorem \ref{th:alg-psym-correct})
		There exists a property $\mathcal{P}$ of distributions over $\{0,1\}^n$ that has a weak $2$-memory $\eps$-test that uses $O(\poly(\eps^{-1}) \log n)$ queries for every $\eps > 0$, but for which there exists some $\eps_0 > 0$ so that any forward-only adaptive $\eps_0$-test requires $\Omega(\poly(n))$ queries.
	\end{theorem*}
	
	To prove the theorem, we define a property that catches the idea of symmetric functions. For some symmetric function $f : [m] \times [m] \to \{0,1\}$, a distribution in the property draws a random key $a \in [m]$ and returns a vector that contains both $a$ (using a high distance code of length $m$) and all values of $f$ at points $(a,b)$ for $b \in [m]$.
	
	For the upper-bound, the algorithm makes a sequence of independent iterations of two samples at a time. In every iteration, it gathers their ``keys'' $a_1$ and $a_2$, verifies the correctness of their codewords, and then checks whether $f(a_1,a_2) = f(a_2,a_1)$. There are some cases that should be carefully analyzed, for example the case where the distribution does not correspond to a single $f$, or the case where some values for ``$a$'' appear very rarely or not at all, but these do not defeat the above algorithm (they somewhat affect its number of needed iterations).
	
	The lower bound follows from a forward-only algorithm being given access to every sample without any knowledge about the keys of ``future'' samples. If the algorithm has only one accessible sample at a time, it can only ``guess'' the other key, but the probability to actually draw a later sample with that key is too low, unless the algorithm collects queries according to about $\sqrt{m}$ guessed keys.
	
	\paragraph{Larger memory generalization}
	We generalize the above theorem to state an exponential separation between the $k$-weak model and the $k-1$-strong model, for every $k \ge 2$:
	\begin{theorem*}(Informal, combined statement of Theorem \ref{th:fwd-lbnd-sym} and Theorem \ref{th:alg-psym-correct})
		For every fixed $k \ge 2$, there exists a property $\mathcal{P}_k$ of distributions over $\{0,1\}^n$ that has a weak $k$-memory $\eps$-test that uses $O(\poly(\eps^{-1}) \log n)$ queries for every $\eps > 0$, but there exists some $\eps_k > 0$ for which any strong $k-1$-memory adaptive $\eps_k$-test requires $\Omega(\poly(n))$ queries.
	\end{theorem*}
	Note that the degree of the polynomials in the above theorem's statement, as well as some hidden constant factors, depend on $k$.
	
	
	We define a property based on parity, which generalizes the above symmetry property. Suppose that $f : \binom{[m]}{k} \to \{0,1\}^k$ is a function such that $f(A)$ has zero parity for every subset $A \subseteq [m]$ of size $k$. We ``encode'' such a function as a distribution, making sure to ``separate'' the $k$ bits of $f(A)$ to $k$ different samples. A typical sample in the distribution would have an encoding (using a high distance code) of a random key $a \in [m]$, followed by some information on $f(A)$ for every $A$ that contains $a$. Specifically, for each such $A$ we supply the $i$th bit of $f(A)$, where $i$ is the ``rank'' of $a$ in $A$ (going by the natural order over $[m]$).
	
	For the upper bound, the algorithm makes a sequence of independent iterations of $k$ samples at a time. In every iteration it gathers the keys $a_1,\ldots,a_k$ and verifies their codewords. If they are all different, the algorithm constructs the value of $f(\{a_1,\ldots,a_k\})$ and checks its parity.
	
	For the lower bound, if the algorithm has less than $k$ accessible samples at a time, again it can only ``guess'' the missing key, and the probability to make the right guess is too low.
	
	We go even further, and show that even if the $k-1$-memory algorithm is allowed to choose which of the samples are retained in every stage (strong $k-1$-memory) rather than keeping a sliding window of recent history, the exponential separation still holds. The separation is achieved for an $\eps_k$-test of the property where $\eps_k = \Theta(1/k)$.
	
	\subsubsection*{Remaining open problems}
	
	It is an open problem whether the weak $k$-memory model is indeed strictly weaker than the strong $k$-memory model (for the same $k$). And if so, is the separation exponential? Also, we do not know whether or not for every $k$ there exists $k^*$ such that the $k^*$-weak model contains the $k$-strong one.
	
	We believe that there exist some $\eps_0 > 0$ and 
	$0 < \alpha < 1$ such that for every sufficiently large $k$, there is an exponential separation between the weak $k$-memory model and the strong $\alpha k$-memory model, with respect to an $\eps_0$-test, rather than the separation for $\eps_k = \Theta(1/k)$ that we show for $k-1$ vs $k$ memory.
	
	Another interesting open problem is whether the fully adaptive model has a simultaneous exponential separation from all fixed $k$-memory models. That is, whether there exists a property $\mathcal{P}$ and some $\eps_0 > 0$ such that $\eps_0$-testing of $\mathcal{P}$ would require $\Omega(\poly(n))$ queries in every $k$-memory model (the polynomial degree possibly depending on $k$), but $\mathcal{P}$ is $\eps$-testable using $O(\log n)$ queries using a fully adaptive algorithm for every fixed $\eps > 0$.
	
	
	
	\section{Additional preliminaries}
	
	The following are some mechanisms and technical lemmas that will aid us throughout the proofs.
	
	\subsection{Property building blocks}
	
	Here we present some useful notions for defining our properties. The following two definitions are used in most of our constructions.
	
	\begin{definition}[Vectorization of functions]
		Let $f : S \to \Sigma^*$ be a function from a (finite) well-ordered set $S$ to strings over a finite alphabet $\Sigma$. For $S' \subseteq S$, we use $\braket{f(i)|i \in S'}$ to denote the concatenation of $f(s)$ for every $s \in S'$, such that the order of concatenation follows the order that is defined for $S$.
	\end{definition}
	
	\begin{definition}[Sample map] \label{def:sample-map}
		Let $P$ be a distribution over $\Omega_1$ and let $f : \Omega_1 \to \Omega_2$ be a function. We define the \emph{sample map} $f(P)$ as the following distribution:
		\[\forall y \in \Omega_2 : \Pr_{f(P)}[y] \overset{\mathrm{def}}{=} \Pr_{x \sim P}[f(x) = y]\]
	\end{definition}
	
	We will also make good use of the following notational conventions.
	
	\begin{definition}[Binomial collection]
		Let $S$ be a set and $k$ be an integer. Define $\binom{S}{k}$ as the set of all subsets of $S$ whose size is exactly $k$.
	\end{definition}
	
	\begin{definition}[Rank] \label{def:ord-a-A}
		Let $A$ be a finite, well ordered set and let $a \in A$. We define $\mathrm{ord}(a,A)$ as the \emph{ranking} of $a$ in $A$. Formally, $\mathrm{ord}(a,A) = \card{\left\{a' \in A \cond a' \le a \right\}}$.
	\end{definition}
	
	From now on we will use $[k]$ to denote $\{1,\ldots,k\}$. In particular for $1\leq i\leq k$ we have $\mathrm{ord}(i,[k])=i$.
	
	\subsection{Reductions between properties} \label{subsec:std-mdl-rdc}
	
	As per Definition \ref{def:sample-map}, given a distribution $P$ over $(\Sigma_1)^n$ and a function $f:(\Sigma_1)^n\to(\Sigma_2)^k$, the sample map $f(P)$ is a distribution over $(\Sigma_2)^k$.

	\begin{lemma} \label{lemma:bndfmap}
		Let $P$ and $Q$ be distributions over some metric set and let $f:(\Sigma_1)^n\to(\Sigma_2)^k$ be a function. If there are two constant factors $0 < a < b$ such that $a \cdot d(x,y) \le d(f(x), f(y)) \le b \cdot d(x,y)$ for every $x, y \in (\Sigma_1)^n$, then $a \cdot d(P,Q) \le d(f(P), f(Q)) \le b \cdot d(P,Q)$.
	\end{lemma}
	
	\begin{proof}
		The upper bound is immediate by taking a transfer distribution $T\in\mathcal{T}(P,Q)$ and moving to the sample map $g(T)\in\mathcal{T}(f(P),f(Q))$, where $g$ is defined by $g(x,y)=(f(x),f(y))$.
		
		For the lower bound, let $T$ be a transfer distribution from $f(P)$ to $f(Q)$. Let $T'$ be the following transfer distribution from $P$ to $Q$:	
		\[T'(x,y) = \frac{\Pr_P[x] \Pr_Q[y]}{\Pr_{f(P)}[f(x)] \Pr_{f(Q)}[f(y)]} T(f(x),f(y))\]
		
		And bound the distance:	
		\begin{eqnarray*}
			a\cdot d(P,Q) &\le & \E_{(x,y) \sim T'}\left[a \cdot d(x,y)\right] \le \E_{(x,y) \sim T'}\left[d(f(x),f(y))\right] \\
			& = & \sum_{u,v \in (\Sigma_2)^k} \Pr_{(x,y) \sim T'}\left[f(x)=u, f(y)=v\right] d(u,v) \\
			& = & \sum_{u,v \in (\Sigma_2)^k} T(u,v) \cdot d(u,v) = \E_{(u,v) \sim T}\left[d(u,v)\right]
		\end{eqnarray*}
		Hence $d(f(P),f(Q)) = \inf_{T \in \mathcal{T}(f(P),f(Q))} \E_{(u,v) \sim T}\left[d(u,v)\right] \ge a \cdot d(P,Q)$.
	\end{proof}
	
	Assume that we have some property of distributions of $n$-length strings over a finite alphabet $\Sigma$ of size $m$, rather than over $\{0,1\}$. Consider some error correction code $C : \Sigma \to \{0,1\}^{2 \log_2 m}$ whose minimal distance is at least $\frac{1}{3}$. We extend $C$ to be defined over $\Sigma^n \to \{0,1\}^{2 n \log_2 m}$ by encoding every element individually.
	\[C(x_1 \ldots x_n) \overset{\mathrm{def}}{=} C(x_1) \ldots C(x_n)\]
	Lemma \ref{lemma:bndfmap} implies that for every $P$ and $Q$ that are distributions over $\Sigma^n$, it holds that
	\[\frac{1}{3} d(P, Q) \le d(C(P), C(Q)) \le d(P,Q)\]
	We note that for all models that we define, using this reduction keeps the algorithm in its respective model, and also preserves one-sided error (if the original algorithm has it).
	
	Based on the above inequality we observe that if there exists an $\eps$-tester for a property over $\Sigma^n$ that uses $s$ samples and $q$ queries, then there exists a $3\eps$-tester for the corresponding binary property, that uses $s$ samples and at most $2 q \log_2 m$ bit queries. Also, if there is no $\eps$-tester for the property over $\Sigma$ (for some $s$ and $q$ bounds), then there is no $\eps$-tester for the encoded property (for the same $s$ and $q$ bounds).
	
	\subsection{Useful Properties}
	
	In the following we will use (very sparse) systematic codes, whose existence is well-known.
	
	\begin{lemma}[Systematic code]
		\label{lemma:syscode}
		There exists a set $\mathcal{C}$ of error correction codes, such that for every $n \ge m \ge 10$, it has a code $C_{m,n} : [m] \to \{0,1\}^n$ with the following properties: (1) Its minimal codeword distance is at least $\frac{1}{3}$ and (2) The projection of $C_{m,n}$ on its first $\ceil{\log_2 m}$ is one-to-one, that is, $C_{m,n}$ can be decoded by reading the first $\ceil{\log_2 m}$ bits.
	\end{lemma}
	From now on, every use of systematic codes refers to the set $\mathcal{C}$ that is guaranteed by Lemma \ref{lemma:syscode}, usually denoted just by $C$ (rather than the explicit notion $C_{m,n}$).
	
	The next property is very useful for proving adaptivity gaps.
	
	\begin{definition}[string property $\mathbf{cpal}$, see \cite{fischer2022}, \cite{akns99}]
		\label{def:prop:2pal}
		For any fixed $n$, the property $\mathbf{cpal}$ is defined over $\{0,1,2,3\}^n$ as the set of $n$-long strings that are concatenations of a palindrome over $\{0,1\}$ and a palindrome over $\{2,3\}$ (in this order).
	\end{definition}
	
	The following lemma is well-known (the adaptive bound, using binary search, is described in \cite{fischer2022}).
	
	\begin{lemma}
		Property $\mathbf{cpal}$ does not have a non-adaptive $\frac15$-test using $o(\sqrt{n})$ queries, while having an adaptive $\eps$-test using $O(\log(n)+1/\eps)$ many queries.
	\end{lemma}
	
	In \cite{fischer2022} this was made into a distribution property by using ``distributions'' that are deterministic.
	
	\begin{definition}[Distribution property $\mathbf{CPal}$, see \cite{fischer2022}]
		\label{def:prop:p-2pal}
		For a fixed, even $n$, the property $\mathbf{CPal}$ is defined as the set of distributions over $\{0,1\}^n$ that are deterministic (have support size $1$), whose support is an element that belongs to $\mathbf{cpal}$, with respect to the encoding $(0,1,2,3) \mapsto (00,01,10,11)$.
	\end{definition}
	
	In Subsection \ref{sec:local:subsec:exp-sep-na} we use $\mathbf{CPal}$ to show an exponential separation between the non-adaptive model and the locally bounded model.
	
	Our next property relies on function inverses to provide adaptivity bounds, and was first investigated in relation to \cite{ergun99}. For a technical reason (that will allow for one-sided error testing later on) we add a special provision for function equality (the original property allowed only for inverse functions).
	
	\begin{definition}[Function property $\mathbf{inv}$]
		\label{def:prop:inv}
		For a fixed $n$, the property $\mathbf{inv}$ is defined over $[n]^{[2n]}$ as the set of ordered pairs of functions $f,g : [n] \to [n]$ such that either $f(i) = g(i)$ for every $1 \le i \le n$ or $g(f(i)) = i$ for every $1 \le i \le n$.
	\end{definition}
	
	It is well-known (first proved in a more general version in \cite{ergun99}) that an $\eps$-test for function inverses takes $O(1/\eps)$ many queries, while e.g.\ testing a single function $f$ for being a bijection requires at least $\Omega(\sqrt{n})$ many queries. For making it into a distribution property we ``split apart'' $f$ and $g$.
	
	\begin{definition}[Distribution property $\mathbf{Inv}$]
		\label{def:prop:p-inv}
		For a fixed $n$, the property $\mathbf{Inv}$ is defined as the set of distributions over $[n]^{[n]}$ that are supported by a set of the form $\{f,g\}$ such that $(f,g) \in \mathbf{inv}$. Note that in particular all deterministic distributions satisfy $\mathbf{Inv}$, since we allow $f=g$ to occur.
	\end{definition}
	
	\begin{definition}[Distribution property $\mathbf{Inv}^*$]
		For a fixed $n$, let $C_n : [n] \to \{0,1\}^{2\ceil{\log_2} n}$ be an error-correction code whose distance is at least $\frac{1}{3}$. We define $\mathbf{Inv}^*$ as the property of distributions over $\{0,1\}^{2\ceil{\log_2 n} n}$, that can be represented as $C_n(P)$ for $P \in \mathbf{Inv}$ (see the discussion after Lemma \ref{lemma:bndfmap}).
	\end{definition}
	
	In Subsection \ref{sec:local:sub:poly-lbnd-p-inv} and Subsection \ref{sec:forward:subsec:exp-sep-local} we use $\mathbf{Inv}$ through its encoding $\mathbf{Inv}^*$ to show an exponential separation between the locally bounded model and the forward-only model.
	
	We finally define a simple property of a matrix (considered as a function with two variables) being symmetric.
	
	\begin{definition}[Matrix property $\mathbf{sym}$]
		\label{def:prop:sym}
		For a fixed $n$, the property $\mathbf{sym}$ of functions with two variables $f:[n]^2 \to \{0,1\}$ is defined as the property of being symmetric, i.e.\ satisfying $f(i,j)=f(j,i)$ for all $i,j\in [n]$.
	\end{definition}
	
	The corresponding distribution property is inspired by considering distributions over the rows of a symmetric matrix, along with properly encoded identifiers.
	
	\begin{definition}[Distribution property $\mathbf{Sym}$]
		\label{def:prop:Sym}
		For any $m$ and the systematic code $C : [m] \to \{0,1\}^m$ from Lemma \ref{lemma:syscode}, the property $\mathbf{Sym}$ is defined as the set of distributions for which \[\Pr\limits_{x \sim P}\left[\exists a \in [m]  : x_{1,\ldots,m} = C(a)\right] = 1\] (i.e.\ all vectors start with an encoding of a ``row identifier''), and for every $a,b \in [m]$, \[\Pr\limits_{x,y \sim P}\left[x_{1,\ldots,m} = C(a) \wedge y_{1,\ldots,m} = C(b) \wedge x_{m+b} \ne y_{m+a}\right] = 0\] (if two ``identifiers'' $a$ and $b$ appear with positive probability, then the respective ``$f(a,b)$'' and ``$f(b,a)$'' are identical).
	\end{definition}
	To understand Definition \ref{def:prop:Sym}, consider first the set of distributions $P$ over $\{0,1\}^{2m}$ that are supported over a set of the form $\left\{C(a),\langle f(a,b)\rangle_{b\in [n]}:a\in [n]\right\}$ where $f$ satisfies $\mathbf{sym}$. However, we need to go in a more roundabout way when defining $\mathbf{Sym}$ due to technical difficulties when only a subset of the possible identifiers appears in the distribution. In Subsection \ref{sec:forward:subsec:poly-lbnd-Sym} and Subsection \ref{sec:kmem:subsec:exp-sep-forward-bounded} we use $\mathbf{Sym}$ to show an exponential separation between the forward-only model and the weak $2$-memory model.
	
	\subsection{Useful lemmas}
	
	The following lemma is well known and is justified by Markov's inequality for $\tilde{X} = 1 - X$.
	
	\begin{lemma}[reverse Markov's inequality]
		\label{lemma:anti-markov}
		Let $X$ be a random variable whose value is bounded between $0$ and $1$. Then for every $0 < \rho < 1$,
		$\Pr[X > \rho\E[X]] \ge (1 - \rho)\E[X]$. Specifically, $\Pr[X > \frac{1}{2}\E[X]] \ge \frac{1}{2}\E[X]$.
	\end{lemma}
	
	
	The following lemma simplifies EMD-distance lower bounds, by characterizing for some properties the distance between them as that achievable by a ``direct translation'' of every vector (or in other words, a sample map). But before the lemma itself we need to define the relevant properties.
	
	\begin{definition}
		Given a family $\Pi$ of subsets of $\Sigma^n$ that is monotone non-increasing, that is, such that for every $A \in \Pi$ and $B \subseteq A$, $B \in \Pi$ too, we define the property $\mathcal{D}(\Pi)=\bigcup_{A\in\Pi}\mathcal{D}(A)$ as the property of having a support that is a member of $\Pi$.
	\end{definition}
	
	\begin{lemma}
		\label{lemma:canonical-dist-from-support-property}
		For a fixed alphabet $\Sigma$, let $\Pi$ be a monotone non-increasing family of subsets of $\Sigma^n$. For every distribution $P \in \mathcal{D}(\Sigma^n)$ there is an $A \in \Pi$ and a function $f : \supp(P) \to A$ such that $d(P,\mathcal{D}(\Pi)) = d(P,f(P)) = \sum_{x \in \supp(P)} \Pr_P[x] d(x,f(x))$.
	\end{lemma}
	\begin{proof}
		Let $Q$ be a distribution that realizes the distance of $P$ from $\mathcal{D}(\Pi)$, so that $\supp(Q) \in \Pi$ and $d(P,\mathcal{D}(\Pi)) = d(P,Q)$. Let $T$ be a transfer distribution (over $\Sigma^n \times \Sigma^n$) that realizes the (EMD) distance between $P$ and $Q$. For every $x \in \supp(P)$, let $f(x) = \arg \min\limits_{y \in \supp(Q)} d(x,y)$ (ties are broken arbitrarily but consistently). Observe that $\supp(f(P)) \subseteq \supp(Q)$, and thus $f(P) \in \mathcal{D}(\Pi)$. Finally,
		\begin{eqnarray*}
			d(P,\mathcal{D}(\Pi)) \le d(P,f(P)) & \le & \sum_{x \in \supp(P)} \Pr_P[x] d(x,f(x)) \\
			& = & \sum_{\substack{x \in \supp(P) \\ y \in \supp(Q)}} T(x,y) d(x,f(x)) \\
			& \le & \sum_{\substack{x \in \supp(P) \\ y \in \supp(Q)}} T(x,y) d(x,y) =\ d(P,Q) = d(P,\mathcal{D}(\Pi))
		\end{eqnarray*}
		and hence all are equal.
	\end{proof}
	
	Finally we state the following ubiquitous lemma for property testing lower bounds. A restricted version appears in \cite{fischer2004art}. A specific instance of this lemma for (fully adaptive) decision trees was first implicitly proved in \cite{fns04}.
	
	\begin{lemma}[useful form of Yao's principle]
		\label{lemma:yao-principle}
		Fix some $\eps > 0$ and $\alpha < \frac{1}{3}$, and let $P$ be a property of distributions over length $n$ strings over a specific alphabet $\Sigma$. Let $D_\mathrm{yes}$ be a distribution over distributions over strings that draws distributions that belong to $P$, and let $D_\mathrm{no}$ be a distribution over distributions over strings that draws a distribution that is $\eps$-far from $P$ with probability $1 - \alpha$ or more. If, for every allowable deterministic algorithm that uses less than $q$ queries, the variation distance between the distribution over answer sequences (e.g.\ leaf identifiers) from an input drawn from $D_\mathrm{yes}$ and the distribution over answer sequences from an input drawn from $D_\mathrm{no}$ is less than $\frac{1}{3} - \alpha$, then every $\eps$-test (in the corresponding model) for $P$ must use at least $q$ queries.
	\end{lemma}
	For the meaning of ``allowable deterministic algorithms'' above, refer to the discussion surrounding the various models of adaptivity defined in Subsection \ref{sec:prelim:subsec:models}. This form of Yao's lemma is well known and justified by analyzing the behavior of any deterministic algorithm over the distribution over input distributions $\frac{1}{2}(D_\mathrm{yes}+D_\mathrm{no})$, leading to an error probability larger than $\frac{1}{3}$.
	
	\section{The non-adaptive model} \label{sec:na}
	
	The following section presents some core properties and methodologies that serve as building blocks for other Huge Object algorithms and their analysis. While the results in this section appear implicitly in \cite{gr2022} (through non-specific reductions), we optimize their query complexity using property-specific algorithms.
	
	\subsection{All zero test}
	
	The all zero test is conceptually the simplest non-trivial testing problem in every reasonable model. Despite its simplicity, in the Huge Object model it is a core building block for reducing the polynomial order of $\eps$ in the query complexity of some properties, compared to black box reductions like in \cite[Theorem 1.4]{gr2022}. Formally, a distribution $P$ over $\{0,1\}^n$ belongs to $\mathcal{D}(0)$ if $\supp(P) = \{0^n\}$.
	
	The $\eps$-testing algorithm is quite simple. We take $\ceil{\eps^{-1}}$ samples, and from each one of them we query a randomly chosen index. We accept if all answers are 0, and otherwise we reject.
	
	\begin{algorithm}
		\caption{One-sided $\eps$-test for all zero, non adaptive, $O(\eps^{-1})$ queries}
		\label{alg:na:p-zero}
		\begin{algorithmic}
			\State \textbf{take} $s = \ceil{\eps^{-1}}$ samples.
			\ForRange{$i$}{$1$}{$s$}
			\State \textbf{choose} $j_i \in [n]$ uniformly at random.
			\State \textbf{query} sample $i$ at index $j_i$, giving $b_i$.
			\If{$b_i \ne 0$}
			\State \Return \reject
			\EndIf
			\EndFor
			\State \Return \accept
		\end{algorithmic}
	\end{algorithm}
	
	\begin{observation} \label{obs:na:p-zero}
		Algorithm \ref{alg:na:p-zero} is a one-sided error $\eps$-test, and its query complexity is $O(\eps^{-1})$.
	\end{observation}
	
	\begin{proof}
		Given a string sample $x$, the probability to query a $1$-bit is exactly $d(x,0^n)$. If $x$ itself is drawn from another distribution $P$, then:
		\[\Pr_{\substack{x \sim P \\ j \sim [n]}}\left[x_j = 1\right] = \E_{x \sim P}\left[\Pr_{j \sim [n]}\left[x_j = 1\right]\right] = \E_{x \sim P}\left[d(x,0^n)\right] = d(P,\mathcal{D}(0))\]
		Where the last transition relies on Lemma \ref{lemma:emd-supp}. For $\eps$-far inputs,
		\[\Pr_P\left[\accept\right] = \prod_{i=1}^{s} \Pr_{\substack{x^i \sim P \\ j \sim [n]}}\left[(x^i)_j = 0\right] = \prod_{i=1}^{s} \left(1 - d(P,\mathcal{D}(0))\right) = (1 - d(P,\mathcal{D}(0)))^s\]
		
		For $s \ge \eps^{-1}$, the probability to accept $\eps$-far inputs is at most $(1-\eps)^{\eps^{-1}} < \frac{1}{2}$ as desired.
	\end{proof}
	
	\subsection{Determinism test}
	
	We show a one-sided $\eps$-test algorithm for the property of having only one element in the support, using $O(\eps^{-1})$ samples and $O(\eps^{-1})$ queries.
	
	Consider some fixed $z \in \{0,1\}^n$. If a distribution is $\eps$-far from being deterministic, then it must also be $\eps$-far from being supported by $\{z\}$. Our algorithm considers the first sample as $z$, and then it uses the other samples to test $P$ for being supported by $\{z\}$. The cost of every logical query is two physical queries (because $z$ is not actually fixed, and to find its individual bits we need to query them).
	
	\begin{algorithm}[H]
		\caption{One-sided $\eps$-test for determinism, non adaptive, $O(\eps^{-1})$ queries}
		\label{alg:na:determinism}
		\begin{algorithmic}
			\State \textbf{take} $s = 1 + \ceil{\eps^{-1}}$ samples.
			\ForRange{$i$}{$2$}{$s$}
			\State \textbf{choose} $j_i \in [n]$, uniformly at random.
			\State \textbf{query} $j_i$ at sample $1$, giving $x^1_{j_i}$.
			\State \textbf{query} $j_i$ at sample $i$, giving $x^i_{j_i}$.
			\If{$x^i_{j_i} \ne x^1_{j_i}$}
			\State \Return \reject
			\EndIf
			\EndFor
			\State \Return \accept
		\end{algorithmic}
	\end{algorithm}
	
	\begin{observation} \label{obs:na:determinism}
		Algorithm \ref{alg:na:determinism} is a non-adaptive one-sided error $\eps$-test for determinism, and its query complexity is $O(\eps^{-1})$.
	\end{observation}
	
	\begin{proof}
		By the discussion above, proving that it is a one-sided $\eps$-test is almost identical to the proof of the all-zero test. Observe that the algorithm cannot reject an input with support size $1$. The probability to reject an $\eps$-far input is:
		\[\Pr_P\left[\reject\right] = \sum_{z \in \{0,1\}^n} \Pr\limits_P\left[x^1 = z\right] \underbrace{\Pr\limits_P\left[\reject \cond x^1 = z \right]}_{\geq 1 / 2\mathrm{\ like\ the\ all\ zero\ test}} \geq \frac{1}{2}\]
		
		To show that it is non-adaptive, note that we can make the random choices for $j_2,\ldots,j_s$ in advance, and then we can query these indexes from $x^1$ and the other corresponding samples in a single batch.
	\end{proof}
	
	\subsection{Bounded support test}
	
	We show a one-sided, non-adaptive $\eps$-test algorithm for the property of having at most $m$ elements in the support, using $O(\eps^{-1} m)$ samples and $O(\eps^{-2} m \log m)$ queries.
	
	\begin{lemma}
		\label{lemma:sparse-support-cover}
		Let $P$ be a distribution over $\{0,1\}^n$ that is $\eps$-far from being supported by $m$ elements (or less). The expected number of independent samples that we have to draw until we get $m+1$ elements of the support that are pairwise $\frac{1}{2}\eps$-far, is at most $1 + 2 \eps^{-1} m$.
	\end{lemma}
	\begin{proof}
		For the purpose of the analysis, consider an infinite sequence $X_1,X_2,\ldots$ of samples that are independently drawn from $P$. For every set $A$ of at most $m$ elements, the expected distance of the next sample from $A$ is at least $\eps$ (since otherwise $P$ would be $\eps$-close to be supported by $A$). By reverse Markov's inequality (Lemma \ref{lemma:anti-markov}), the probability to draw a sample that is $\frac{1}{2}\eps$-far from $A$ is at least $\frac{1}{2}\eps$.
		
		For every $i \le 1$, let $T_i$ be the index of the first sample that is $\frac{1}{2}\eps$-far from $\{X_{T_1},\ldots,X_{T_{i-1}}\}$. Trivially, $T_1 = 1$, and for every $2 \le i \le m+1$, $T_i - T_{i-1}$ is a geometric variable with success probability of at least $\frac{1}{2}\eps$ (the set $\{X_{T_1},\ldots,X_{T_{i-1}}\}$ takes the role of the set $A$ in the discussion above), and thus its expected value is at most $2\eps^{-1}$.
		
		By linearity of expectation, $\E\left[T_{m+1}\right] = \E\left[T_1\right] + \sum_{i=2}^{m+1} \E\left[T_i - T_{i-1}\right] \le 1 + 2 m \eps^{-1}$.
	\end{proof}
	
	The algorithm works as follows: we choose a set $J$ of $t = \ceil{4 \eps^{-1} \ln m}$ indexes, and take $s = 1 + \ceil{8 m \eps^{-1}}$ samples. Then we query every sample in all indexes of $J$, and reject if we find a set of $m+1$ samples whose restrictions to $J$ are distinct.
	
	\begin{algorithm}
		\caption{One sided $\eps$-test for $m$-bounded support, non adaptive, $O(\eps^{-2} m \log m)$ queries}
		\label{alg:na:m-support}
		
		\begin{algorithmic}
			\State \textbf{take} $s = 1 + \ceil{8 \eps^{-1}m}$ samples.
			\State \textbf{let} $t = \ceil{4 \eps^{-1} (\ln m + 2)}$
			\State \textbf{choose} $j_1,\ldots,j_t \in [n]$ uniformly and independently at random.
			\State \textbf{let} $J = \{j_1, \ldots, j_t\}$
			\ForRange{$i$}{$1$}{$s$}
			\State \textbf{query} sample $i$ at $j$ for every $j \in J$, giving substring $y^i$ of length $\card{J}$.
			\EndFor
			\If{$\card{\set{y^1,\ldots,y^s}} > m$}
			\State \Return \reject
			\EndIf
			\State \Return \accept
		\end{algorithmic}
	\end{algorithm}
	
	\begin{theorem} \label{th:na:m-support}
		Algorithm \ref{alg:na:m-support} is a one-sided $\eps$-test for being supported by at most $m$ elements.
	\end{theorem}
	
	\begin{proof}
		For proving complexity, observe that the algorithm draws $O(\eps^{-1} m)$ samples and makes $O(\eps^{-1} \log m)$ queries to each of them, giving a total of $O(\eps^{-2} m \log m)$ queries.
		
		For perfect completeness, consider an input distribution $P$ that is supported by a set of $k$ elements (for $k \le m$).Note that in this case $\card{\{y^1,\ldots,y^s\}} \le k$ for every choice of $J$. Thus, the algorithm must accept it with probability $1$.
		
		For soundness, consider an input distribution $P$ that is $\eps$-far from being supported by any set of $m$ elements. By Lemma \ref{lemma:sparse-support-cover} and Markov's inequality, with probability higher than $1 - \frac{1}{4}$, there are at least $m+1$ pairwise $\frac{1}{2}\eps$-far elements within the $s$ samples of the algorithm. If this happens, then for every pair of these elements, the probability that they agree on all indexes of $J$ is at most $(1 - \frac{1}{2}\eps)^{4\eps (\ln m + 2)}$, which is less than $\frac{1}{e^2 m^2}$. The probability that $J$ fails to distinguish even one of the $\binom{m}{2}$ pairs is at most $e^{-2}$. Hence, the probability of the algorithm to reject is at least $1 - \frac{1}{4} - e^{-2} > \frac{1}{2}$.
	\end{proof}
	
	\section{The locally-bounded model} \label{sec:local}
	
	The locally bounded model captures the concept of distributed execution in the Huge Object model. Every sample is processes adaptively using a (possibly) different logic, but nothing is shared across samples. After all nodes are done, the algorithm makes its decision based on the concatenation of their results. Lower bounds for this model are surprisingly hard to prove, and we use a corresponding string model to show them.
	
	\subsection{Split adaptive string testing}
	
	We define a model of string algorithms that helps us to analyze some variants of locally bounded adaptive algorithms.
	\begin{definition}[Split adaptive algorithm]
		For a fixed $k$, a \emph{$k$-split adaptive} deterministic algorithm for $n$-long strings (where $n$ is divisible by $k$) over some alphabet $\Sigma$ is a sequence of $k$ decision trees $T_1,\ldots,T_k$, where the tree $T_i$ can only query at indexes between $(i-1)k+1$ and $ik$, and a set of accepted answer sequences. The query complexity of the algorithm is defined as the sum of heights of its trees.
	\end{definition}
	
	\begin{observation}
		Every $k$-split deterministic adaptive algorithm can be represented as the tuple $(T_1,\ldots,T_k,A)$, that consists of its $k$ decision trees and the set of accepted answer sequences. A $k$-split probabilistic algorithm can be seen as a distribution over such tuples.
	\end{observation}
	
	\begin{lemma}[Construction of a $2$-split adaptive string algorithm from locally bounded one]
		\label{lemma:2-split-construction} 
		For some fixed alphabet $\Sigma$, let $R \subseteq \left(\Sigma^n\right)^2$ be a reflexive and symmetric binary relation, and let $\Pi_R$ be the property of $2n$-long strings that are concatenation of $u,v \in \Sigma^n$ such that $(u,v) \in R$. Let $\mathcal{P}_R$ be the property over distributions over $\Sigma^n$ which states that $P\in\mathcal{P}_R$ if it is supported over a set $\{u,v\}$ such that $(u,v) \in R$ (note that by the assumption that $R$ is reflexive, every deterministic distribution is in $\mathcal{P}_R$). There exists an algorithmic construction whose input is a $2$-split adaptive $\eps$-test algorithm for $\Pi_R$, and its output is a locally-bounded $\eps$-test algorithm for $\mathcal{P}_R$, with the same number of queries. Also, the construction preserves one-sided error, if exists in the input.
	\end{lemma}
	
	There is a natural generalization of lemma \ref{lemma:2-split-construction} for every fixed $k \ge 2$, but it is much more detailed, and we choose to avoid it as we only need the $k=2$ case.
	
	\begin{proof}
		For every $u,v \in \Sigma^n$, let $P_{u,v}$ be the distribution that draws $u$ with probability $\frac{1}{2}$ and $v$ with probability $\frac{1}{2}$ (if $u=v$ then $P_{u,v}$ is deterministic). Observe that $d(P_{u,v}, \mathcal{P}_R) = d(uv, \Pi_R)$:
		if $u=v$ then they are both $0$ (because $R$ is reflexive), otherwise
		\[d(P_{u,v}, \mathcal{P}_R) \le \min\limits_{(u^*,v^*) \in R} \left( \frac{1}{2} d(u,u^*) + \frac{1}{2} d(v,v^*) \right) = \min\limits_{u^* v^* \in \Pi_R} {d(uv, u^* v^*)} = d(uv, \Pi_R)\]
		On the other hand, by Lemma \ref{lemma:canonical-dist-from-support-property} there exists $(u^*,v^*) \in R$ such that:
		\[d(P_{u,v}, \mathcal{P}_R) = \frac{1}{2}d(u,u^*) + \frac{1}{2}d(v,v^*) = d(uv, u^* v^*) \ge d(uv, \Pi_R)\]
		Let $\mathcal{A}$ be a locally bounded $\eps$-test for $\mathcal{P}_R$. $\mathcal{A}$ is a distribution over deterministic algorithms of the form $(T_1,\ldots,T_s ; A)$. Consider the following \emph{conceptual} algorithm for strings: the input is a $2n$-long string $uv$ (where $u$ is the $n$-prefix and $v$ is the $n$-suffix). Simulate $\mathcal{A}$ with $P_{u,v}$ as its input, and return the same answer. If $uv \in \Pi_R$ then $P_{u,v} \in \mathcal{P}_R$, and the algorithm should accept with probability higher than $\frac{2}{3}$ (observe that one-side error is preserved). If $uv$ is $\eps$-far from $\Pi_R$ then $d(P_{u,v}, \mathcal{P}_R) = d(uv, \Pi_R) > \eps$, and the algorithm should reject with probability higher than $\frac{2}{3}$, as desired.
		
		To complete the proof we show the actual implementation $\tilde{\mathcal{A}}$ of the conceptual algorithm. Recall that a $2$-split adaptive probabilistic algorithm is a distribution over deterministic algorithms of the form $(\tilde{T_1},\tilde{T_2},\tilde{A})$. We draw a deterministic algorithm $(T_1,\ldots,T_s;A)$ from $\mathcal{A}$, and uniformly and independently draw $b_1,\ldots,b_s \in \{0,1\}$. We define the first tree ($\tilde{T}_1$, which is then executed on the $u$-part of the input) as the concatenation of all trees $T_i$ where $b_i = 0$. The second tree ($\tilde{T}_2$, which is then executed on the $v$-part of the input) is defined as the concatenation of all trees $T_i$ where $b_i = 1$ and every query $j$ of the original trees is translated into a query $n+j$ (because we want to query the $v$-part, whose indexes are $n+1,\ldots,2n$). The set of accepted answer sequences $\tilde{A}$ is defined analogously: a pair of leaves in $\tilde{T}_1,\tilde{T_2}$ is accepting if the corresponding sequence of leaves in $T_1,\ldots,T_s$ is an accepting superleaf in $A$.
		
		Every accepting (respectively, rejecting) run of $\tilde{\mathcal{A}}$ given an input $uv \in \Sigma^{2n}$ corresponds to an accepting (respectively, rejecting) run of $\mathcal{A}$ given the input $\mathcal{P}_{u,v}$ that has the same probability to be executed, hence the construction is correct.
	\end{proof}
	
	\subsection{Exponential separation from the non-adaptive model} \label{sec:local:subsec:exp-sep-na}
	
	As mentioned in \cite{fischer2022}, based on a similar analysis in \cite{akns99}, the property $\mathbf{CPal}$ (Definition \ref{def:prop:p-2pal}) requires at least $\Omega(\sqrt{n})$ queries for a non-adaptive $\eps$-test (for sufficiently small values of $\eps$) in the Huge Object model, but can be $\eps$-tested adaptively using $O(\poly\left(\eps^{-1}\right)) \cdot \log n$ queries. Their proof is based on an algorithm that considers every sample individually (for every sample they make an adaptive $O(\eps)$-test for being in $\mathbf{cpal}$), and thus it is locally bounded.
	
	$\mathbf{CPal}$ demonstrates an exponential separation of the locally bounded model and the completely non-adaptive one.
	
	\subsection{Polynomial lower bound for $\mathbf{Inv}$} \label{sec:local:sub:poly-lbnd-p-inv}
	
	\begin{theorem} \label{th:lbnd-local-Inv}
		Every locally-bounded adaptive $\frac{1}{5}$-test for $\mathbf{Inv}$ must make at least $\frac{1}{3}\sqrt{n}$ queries.
	\end{theorem}
	By Lemma \ref{lemma:2-split-construction}, this follows immediately from the following lemma:
	
	\begin{lemma}\label{lemma:lbnd-2-split-inv}
		Every $2$-split adaptive $\frac{1}{5}$-test for $\mathbf{inv}$ must make at least $\frac{1}{3}\sqrt{n}$ queries.
	\end{lemma}
	
	For convenience we denote the identity permutation over $[n]$ by $\mathrm{id}$. We assume that $n > 60$. Let $D_\mathrm{yes}$ be a distribution that chooses some permutation $f : [n] \to [n]$ uniformly at random, and returns $(f,f^{-1})$. Let $D_\mathrm{no}$ be a distribution that chooses two permutations $f,g : [n] \to [n]$ uniformly at random and independently, and returns $(f,g)$. Observe that $D_\mathrm{yes}$ returns only strings in $\mathbf{inv}$.
	
	\begin{lemma}\label{lemma:fg-far}
		For every $n > 60$, $D_\mathrm{no}$ draws a $\frac{1}{5}$-far input with probability more than $1 - \frac{1}{12}$.
	\end{lemma}
	
	\begin{proof}
		The expected distance of $(f,g)$ from being the same function is $\frac{1}{2} - \frac{1}{2n}$, because for every $1 \le i \le n$, the probability that $g(i) \ne f(i)$ is $1 - \frac{1}{n}$. By Markov's inequality (denoting by ``$\mathrm{same}$'' the property of all pairs $(f,g)$ for which $f=g$),
		\[\Pr\left[d((f,g),\mathrm{same}) < \frac{1}{5}\right] = \Pr\left[\frac{1}{2} - d((f,g),\mathrm{same}) > \frac{1}{2} - \frac{1}{5}\right] < \frac{1/n}{3/10} < \frac{1}{18}\]
		
		Let $\tilde{f}$ and $\tilde{g}$ be co-inverse permutations such that $d((f,g),(\tilde{f},\tilde{g})) = d((f,g),\mathbf{inv})$. Observe that $d(f \circ g, \tilde{f} \circ \tilde{g}) \le d(f, \tilde{f}) + d(g, \tilde{g})$, hence $d((f,g),\mathbf{inv}) = \frac{1}{2}d(f,\tilde{f}) + \frac{1}{2}d(g,\tilde{g}) \ge \frac{1}{2}d(f \circ g, \mathrm{id})$. Also, for $f$ and $g$ that are drawn from $D_\mathrm{no}$, the composition $f \circ g$ distributes uniformly, and its expected distance from the identity permutation is $\frac{n-1}{n}$, that is, $\E\left[1 - d(f \circ g,\mathrm{id})\right] = \frac{1}{n}$. By the union bound with the case that $(f,g)$ is $\frac{1}{5}$-close to being deterministic,
		
		\begin{eqnarray*}
			\Pr\left[d\left(\left(f,g\right), \mathrm{inv}\right) < \frac{1}{5}\right]
			\le \frac{1}{18} + \Pr\left[d\left(f \circ g, \mathrm{id}\right) < \frac{2}{5}\right]
			&=& \frac{1}{18} + \Pr_{h \sim \pi([n])} \left[d\left(h, \mathrm{id}\right) < \frac{2}{5}\right] \\
			&=& \frac{1}{18} + \Pr_{h \sim \pi([n])} \left[1 - d\left(h, \mathrm{id}\right) > \frac{3}{5}\right] \\
			&\le& \frac{1}{18} + \frac{1 / n}{3 / 5} = \frac{1}{18} + \frac{5}{3n} < \frac{1}{12}
		\end{eqnarray*}
		
	\end{proof}
	
	Recall Yao's principle, as detailed in Lemma \ref{lemma:yao-principle}. If for every deterministic algorithm that uses less than $q$ queries the variation distance between $D_\mathrm{yes}$ and $D_\mathrm{no}$ is less than $\frac{1}{3} - \frac{1}{12} = \frac{1}{4}$, then every probabilistic $\frac{1}{5}$-tester for $\mathbf{inv}$ must use at least $q$ queries.
	
	We will prove that the lower bound holds even if we have an additional promise on the input, that both $f$ and $g$ are permutations (but not necessarily inverses). Note that our two input distributions satisfy this.
	
	Fix some deterministic algorithm $(T_1, T_2, A)$, and let $q_f$ and $q_g$ be the number of queries in the first and the second tree respectively (so $q_f + q_g = q$). Without loss of generality, we assume that both $T_1$ and $T_2$ are balanced (all leaves of a tree have the same depth), and that every internal node in the $i$th depth has $n-i$ children corresponding to the elements in $\{1,\ldots,n\}$ that did not appear earlier in the path from the root (whose depth is $0$). These trees can handle every sequence of answers when the input is guaranteed to be a pair of permutations, which we assume from now on as per the discussion above. Also, without loss of generality, we assume that the tree never makes a query that it has already made earlier in the path.
	
	From now on, given $(f,g)$ that is drawn according to our input distribution (either $D_\mathrm{yes}$ or $D_\mathrm{no}$), we denote by $T_1(f)$ the path followed by $T_1$ on the input $f$, and analogously denote by $T_2(g)$ the path followed by $T_2$ on $g$. The following lemma, together with Lemma \ref{lemma:fg-far}, immediately implies Lemma \ref{lemma:lbnd-2-split-inv}.
	
	\begin{lemma}\label{lemma:garbled-2-split}
		Given a distribution over inputs $D$, denote by $\tilde{D}$ the resulting distribution over the pair of tree paths $(T_1(f),T_2(g))$ where $(f,g)$ is drawn by $D$. If $q<\frac13\sqrt{n}$, then $d(\tilde{D}_\mathrm{yes},\tilde{D}_\mathrm{no})<\frac18$.
	\end{lemma}
	
	We now define and analyze features in $T_2$ that depend on a fixed path $F$ in $T_1$. Let $a_1,\ldots,a_{q_f}$ be the $f$-queries and let $c_1,\ldots,c_{q_f}$ be their answers ($c_i = f(a_i)$). We define the following \emph{traps} in $T_2$:
	\begin{itemize}
		\item A \emph{revealing node} is an internal node whose query $b$ belongs to $\{c_1,\ldots,c_{q_f}\}$. Observe that the count of revealing nodes depends on the choice of the $f$-path. A \emph{revealing path} is a path that contains at least one revealing node.
		\item A \emph{wrong node} is a node (possibly a leaf) whose parent edge's label belongs to $\{a_1,\ldots,a_{q_f}\}$. Observe that every internal node has exactly $q_f$ children that are wrong, regardless of the choice of $f$. A \emph{wrong path} is a path that contains at least one wrong node.
	\end{itemize}
	
	Note that the above definitions depend only on $F$, that is, if $T_1(f)=T_1(f')=F$, then the sets of revealing nodes and wrong nodes are identical.
	
	\begin{lemma}
		In $T_2$, for any $F$, at most $\frac{q_g q_f}{n}n^{q_g}$ paths are wrong.
	\end{lemma}
	\begin{proof}
		Fix some $f$-path and then choose a path $x_1,\ldots,x_{q_g}$, uniformly at random. The path is wrong if there exists some $1 \le i \le q_g$ such that $x_i \in \{a_1,\ldots,a_{q_f}\}$. For every individual $i$, the probability to do that is at most $\frac{q_f}{n - i + 1}$, hence by the union bound the probability to have a wrong path is at most $\sum_{i=1}^{q_g}\frac{q_f}{n+1-i}$. There are $\prod_{j=0}^{q_g-1}(n-j)$ paths at all, so the total number of wrong paths is bounded by $\sum_{i=1}^{q_g}\frac{q_f}{n+1-i}\prod_{j=0}^{q_g-1}(n-j)\leq\frac{q_gq_f}{n}n^{q_g}$.
	\end{proof}
	
	\begin{lemma}
		The expected number of revealing paths in $T_2$, where $F$ is drawn by taking a uniformly random permutation $f$ and setting $F=T_1(f)$, is at most $\frac{q_g q_f}{n}\prod_{j=0}^{q_g-1}(n-j)\leq\frac{q_g q_f}{n}n^{q_g}$.
	\end{lemma}
	\begin{proof}
		For a uniform choice of $f$, the set $\{c_1,\ldots,c_{q_f}\}$ distributes uniformly over subsets of size $q_f$, and thus every individual internal node is a revealing node with probability $\frac{q_f}{n}$. In every individual path, the expected number of revealing nodes is bounded by $\frac{q_g q_f}{n}$, and thus the probability that it is a revealing path is at most $\frac{q_g q_f}{n}$. By linearity of expectation, the expected number of revealing paths is at most $\frac{q_g q_f}{n} n^{q_g}$.
	\end{proof}
	
	\begin{lemma}
		The expected number of bad (revealing or wrong) paths in $T_2$ is at most $\frac{q^2}{n}n^{q_g}$.
	\end{lemma}
	\begin{proof}
		The sum of the lemmas above bounds the expected number by at most $\frac{2 q_g q_f}{n} n^{q_g}\leq \frac{q^2}{n}n^{q_g}$.
	\end{proof}
	
	\begin{lemma}\label{lemma:no-trap}
		The probability of hitting a trap, under both $D_\mathrm{yes}$ and $D_\mathrm{no}$, is bounded by $e^{\lfrac{q^2}{\left(n - q\right)}} \frac{q^2}{n}$.
	\end{lemma}
	
	\begin{proof}
		Let $N$ be a random variable for the number of bad paths in $T_2$, and let $B$ be the set of traps in $T_2$ that have no ancestor traps. Recall that both $N$ and $B$ depend on the $T_1(f)$. When we want to refer to them conditioned on a certain path $T_1(f)=F$, we denote them by $N(F)$ and $B(F)$ respectively.
		
		The set of bad paths is a disjoint-union of all subtrees of $B$-nodes. For every node $u \in B$ we denote its depth with $l(u)$. The number $N$ of bad paths is at most $\sum_{u \in B} n^{q_g-l(u)}$, because the number of leaves in a subtree of an $l$-deep node is $n^{q_g-l}$. Every bad path must go through exactly one $B$-node, and thus the probability to hit a bad path is the probability to hit an $B$-node. By a disjoint union bound,
		\begin{eqnarray*}
			\Pr\left[\textsc{trap} \cond T_1(f)=F \right] &<&
			\sum_{u \in B(F)} \frac{1}{\left(n - q\right)^{l(u)}} \\
			&\le& \sum_{u \in B(F)}  \frac{e^{\lfrac{q^2}{\left(n - q\right)}}}{n^{l(u)}} \\
			&=& e^{\lfrac{q^2}{\left(n - q\right)}} n^{-q_g} \sum_{u \in B(F)} n^{q_g - l(u)}
			\le e^{\lfrac{q^2}{\left(n - q\right)}} n^{-q_g} N(F)
		\end{eqnarray*}
		To get the first (leftmost) bound, observe that as long as we did not hit a trap so far, the next step distributes uniformly over the set of children that are not eliminated for trivial reasons. For $D_\mathrm{no}$, it is the set of all children (exactly $n-q_g$), and for $D_\mathrm{yes}$  it is the set of non-wrong children (at least $n-q$, noting that we assumed that the current node is not already revealing). To hit a specific trap (that has no ancestor traps), we should take the correct edge $h(u)$ times, and the probability to do that is at most $\frac{1}{n-q}$ for every step. Overall, using Observation \ref{apx:tbnd1}, the probability to hit a trap is:
		\begin{eqnarray*}
			\Pr\left[\textsc{trap}\right]
			&=& \sum_F { \Pr\left[T_1(f)=F\right]\cdot \Pr\left[\textsc{trap} \cond T_1(f)=F \right] } \\
			&\le& \sum_{F} {\Pr\left[T_1(f)=F\right]\cdot e^{\lfrac{q^2}{\left(n - q\right)}} n^{-q_g} N(F) } \\
			&=& e^{\lfrac{q^2}{\left(n - q\right)}} n^{-q_g} \E\left[N(T_1(f))\right]
			\le e^{\lfrac{q^2}{\left(n - q\right)}} \frac{q^2}{n}
		\end{eqnarray*}
		This completes the proof.
	\end{proof}
	
	\begin{proof}[Proof (of Lemma \ref{lemma:garbled-2-split})]
		Both in $D_\mathrm{yes}$ and $D_\mathrm{no}$, the $T_1(f)$ distributes uniformly over the set of $\frac{(n-q)!}{n}$ allowable paths. Conditioned on a specific path $F=T_1(f)$, the path $T_2(g)$ distributes under $D_\mathrm{no}$ uniformly as well. On under $D_\mathrm{yes}$, the distribution of the path $T_2(g)$ (conditioned on $T_1(f)=F$) is more complicated, but when additionally conditioned on avoiding traps, $T_2(g)$ indeed distributes uniformly over the trap-free paths in $T_2$. The number of bad paths in $T_2$ is $\frac{n!}{(n-q_g)!}\Pr\left[\textsc{trap} \cond T_1(f)=F \right]$. Hence the distance between $\tilde{D}_\mathrm{yes}$ and $\tilde{D}_\mathrm{no}$, conditioned on $T_1(f)=F$, is bounded by $\Pr\left[\textsc{trap} \cond T_1(f)=F \right]$. The unconditional variation distance is bounded by the expectation of the conditional variation distance over the choice of $F$, and hence, using Lemma \ref{lemma:no-trap} \[
		d(\tilde{D}_\mathrm{yes},\tilde{D}_\mathrm{no})\leq\sum_F { \Pr\left[T_1(f)=F\right]\cdot \Pr_{D_\mathrm{yes}}\left[\textsc{trap} \cond T_1(f)=F \right]}=\Pr_{D_\mathrm{yes}}\left[\textsc{trap}\right] \le e^{\lfrac{q^2}{\left(n - q\right)}} \frac{q^2}{n} \]
		
		With $q < \frac{1}{3}\sqrt{n}$, the variation distance between $\tilde{D}_\mathrm{yes}$ and $\tilde{D}_\mathrm{no}$ is less than $\frac{1}{8}$, as desired.
	\end{proof}
	
	\section{The forward-only model} \label{sec:forward}
	
	The power of forward-only algorithms over locally bounded ones is the ability to make queries to samples based on knowledge about prior samples. More precisely, it has the advantage of considering relations between samples. We show a forward-only improved support test, that avoids redundant queries that appear to be inevitable in models that do not allow ``communication'' between different samples \cite{supp23}. We also show a logarithmic test for $\mathbf{Inv}$, that demonstrates the ability to test binary relations between samples. However, this power is bounded, because looking back may be necessary to deal with complicated binary relations between the members of a set of an unbounded size, as we show in the lower bound for $\mathbf{Sym}$.
	
	\subsection{Query foresight}
	
	Query foresight a method for constructing a forward-only algorithm out of an unrestricted adaptive one by reordering its queries (possibly with some cost) to avoid ``looking back''. There are some algorithms that could be forward-only but are ``loosely written'', in a sense that they make their queries in a ``needlessly complex'' ordering.
	
	The idea is straightforward: we simulate the run of an adaptive algorithm. Every time that the simulation is about to query a new sample, we make additional speculative queries in the current sample, before dropping it as per the requirement of a forward-only algorithm. If the simulated algorithm makes a query to an old sample, we feed it with the answer of the corresponding speculative query. If such a speculative query does not exists, we either accept (for one-sided algorithms) or behave arbitrarily (for two-sided algorithms). If the prediction is conservative, that is, the speculated queries are ensured to cover all queries to past samples, then the construction guarantees the exact acceptance probability for every individual input. This is not guaranteed when the prediction is not conservative, and in this case we need to analyze the effect of bad speculations.
	
	\subsubsection*{Improved bounded support test}
	
	We show a one-sided error $\eps$-testing algorithm for the property of having at most $m$ elements in the support, using $O(\eps^{-1})$ samples and $O(\eps^{-1} m^2)$ queries, which is more efficient, for a fixed $m$, than what we do in the non-adaptive model. We reduce queries by using the ability to have extra queries only for samples that are found to be ``new'' (at most $m+1$ of them). This algorithm also reports distinct elements from the support as soon as it encounters them. Note that the algorithm is not necessarily optimal. We introduce it to demonstrate query foresight.
	
	The algorithm is generally based on the non-adaptive support test in Section \ref{sec:na}. For an input distribution that is $\eps$-far from being supported by $m$ elements, $O(m\eps^{-1})$ samples are sufficient for having $m+1$ elements that are pairwise $\frac{1}{2}\eps$-far. In the non-adaptive algorithm we just choose a large set of indexes that deals with all pairs of elements at once. We reduce the number of queries by being adaptive. If we already know that some samples are similar to each other, and we draw another one, we only compare it to one of the similar samples rather than to all of them.
	
	Consider an $\eps$-far input distribution, and assume that we have already found $r$ distinct samples ($r \le m$). The expected distance of the next sample from all of them is at least $\eps$ (since otherwise the distribution is not $\eps$-far from being $m$-supported). For every two samples that we know to be different, we also know about a specific query that indicates this. If we query a new sample at all of these indices (at most $r-1$), we immediately find at least $r-1$ samples that are not the same as the new one. As for the last standing sample, we just query it in a uniform index to compare it to the new sample. The probability that they are different is at least $\eps$. The expected number of samples that we have to draw until we find a new distinct sample is then at most $\eps^{-1}$. The expected number of samples for finding $m+1$ distinct samples is hence $1 + m \eps^{-1}$. By Markov's inequality, after $1 + 2 \eps^{-1} m$ many samples, the probability to find $m+1$ distinct samples is at least $\frac{1}{2}$. See Algorithm \ref{alg:s-m+1-mem:m-supp}, whose formal correctness is given below, and observe that Algorithm \ref{alg:forward-m-supp}, which is constructed using speculative queries and is forward only, is logically equivalent.
	
	\begin{algorithm}
		\caption{One sided $\eps$-test for $m$-bounded support, strong $m+1$-memory, $O(\eps^{-1} m^2)$ queries}
		\label{alg:s-m+1-mem:m-supp}
		\begin{algorithmic}
			\State \textbf{Memory storage for samples:} $z^1,\ldots,z^m; x$, all initialized to $\mathbf{NULL}$.
			\State \textbf{Extra cell:} We have another syntactic ``write-only'' memory storage $z^{m+1}$ which we never query.
			\State \textbf{take} $s = 1 + \ceil{2 \eps^{-1} m}$ samples.
			\State \textbf{set} $c,t \gets 0$.
			\State \textbf{set} $j_1, \ldots, j_m \gets \textbf{NULL}$
			\ForRange{$k$}{$1$}{$s$}
			\State \textbf{Invariant 1} $c=m$ or $z^{c+1} = \textbf{NULL}$.
			\State \textbf{Invariant 2} for $1 \le i \le c$, $z^i_J$ are distinct where $J = \{j_1,\ldots,j_t\}$.
			\State \textbf{store} $x \gets \text{sample } k$.
			\State \textbf{query} $x$ at $j_1,\ldots,j_t$, giving substring $y^k$.
			\ForRange{$i$}{$1$}{$c$}
			\State \textbf{query} sample $z^i$ at $j_1,\ldots,j_t$ giving substring $y^i$. \Comment{the $y^i$s are distinct}
			\EndFor
			
			\State \textbf{choose} $j \in [n]$ uniformly at random.
			\State \textbf{query} $x$ at $j$, giving $x_j$.
			
			\If{$\exists i : y^i = y^k$}\Comment{if exists it is unique}
			\State \textbf{query} sample $z^i$ at $j$ giving $z^i_j$.
			\If{$x_j \ne z^i_j$}
			\State \textbf{store} $z^{c+1} \gets x$.
			\State \textbf{set} $j^{t+1} \gets j$.
			\Comment{keep Invariant 2}
			\State \textbf{set} $t \gets t+1$ and $c \gets c+1$. \Comment{keep Invariant 1}
			\EndIf
			\Else
			\State \textbf{store} $z^{c+1} \gets x$. \Comment{Invariant 2 still holds}
			\State \textbf{set} $c \gets c + 1$. \Comment{keep Invariant 1}
			\EndIf
			\If {$c > m$}
			\State \Return \reject
			\EndIf
			\EndFor
			\State \Return \accept
		\end{algorithmic}
	\end{algorithm}
	
	\begin{theorem}
		Algorithm \ref{alg:s-m+1-mem:m-supp} is a one-sided $\eps$-test for being supported by at most $m$ elements.
	\end{theorem}
	
	\begin{proof}
		The query complexity of Algorithm \ref{alg:s-m+1-mem:m-supp} is $O(\eps^{-1} m^2)$: every sample is queried in at most $t+1$ indexes (where $t \le m$) when it is the current one, and in every iteration there is at most one extra query to one of the $z^i$s. The total number of queries is at most $(m+1)s + s$, which is $\Omega(\eps^{-1} m^2)$.
		
		For perfect completeness, observe that the invariants guarantee that $\{z^1,\ldots,z^m\} \setminus \{\mathbf{NULL}\}$ is a subset of the input distribution's support, and thus the algorithm can never reject an input that is supported by at most $m$ elements (it never finds a sample that is different from all $z^i$s).
		
		For soundness, consider an $\eps$-far input distribution, and assume that we have an infinite number of samples. We will bound for every $z^i$ the expected iteration that assigns it with a valid sample. Let $T_1,\ldots,T_m,T_{m+1}$ be these counts. $T_{m+1}$ is the iteration of reject. Trivially, $\Pr[T_1 = 1] = 1$, because we must assign $z^1$ in the first iteration. For $2 \le i \le m+1$, observe that the expected distance of ``the next sample'' from $\{z^1,\ldots,z^{i-1}\}$ is at least $\eps$. Otherwise, the input would be $\eps$-close to be supported by $\{z^1,\ldots,z^{i-1}\}$. Thus:
		
		\begin{eqnarray*}
			\Pr_{x \sim P}\left[y^k \notin \left\{y^1,\ldots,y^t\right\} \lor x_j \ne z^i_j \right] &\ge&
			\min_{1 \le i \le t} \Pr_{x \sim P}\left[x_j \ne z^i_j \right] \\
			&=& \min_{1 \le i \le t} \E_{x \sim P}\left[d\left(x, y^i\right)\right] \ge \E_{x \sim P}\left[\min_{1 \le i \le t} d\left(x, y^i\right)\right] \ge \eps
		\end{eqnarray*}
		
		The number of iterations until $z^i$ is assigned (counting since the assignment of $z^{i-1}$) is a geometric variable with success probability at least $\eps$, and thus its expected value is at most $\eps^{-1}$. By linearity of expectation,
		\[\E\left[T_{m+1} - T_1\right] =\sum_{i=2}^{m+1} \E\left[T_{i+1} - T_{i}\right] \le \eps^{-1}m\]
		By Markov's inequality,	$\Pr\left[T_{m+1} - T_{1} \le 2 \eps^{-1} m \right] > \frac{1}{2}$. Going back to Algorithm \ref{alg:s-m+1-mem:m-supp}, note that it uses $1 + \ceil{2 \eps^{-1}}$ samples (rather than an infinitely many), and thus it rejects every $\eps$-far input with probability higher than $\frac{1}{2}$.
	\end{proof}

	\subsubsection*{Applying query foresight on the improved $m$-support test}
	
	Observe that Algorithm \ref{alg:s-m+1-mem:m-supp} is not forward-only, because it holds up to $m$ samples in memory and keeps querying them with every new sample. Though easier to analyze for correctness, it is not streamlined. We know that every ``sample in memory'' is going to be queried in at most one ``new'' location per each incoming sample, hence we can just choose all these indexes in advance and make all queries as soon as we (virtually) ``store'' the sample in memory. The cost is at most $m$ extra queries for every sample that we take in future. Algorithm \ref{alg:forward-m-supp} is a rephrasing of Algorithm \ref{alg:s-m+1-mem:m-supp} using query foresight.
	
	\begin{algorithm}
		\caption{One sided $\eps$-test for $m$-bounded support, forward only, $O(\eps^{-1} m^2)$ queries}
		\label{alg:forward-m-supp}
		\begin{algorithmic}
			\State \textbf{take} $s = 1 + \ceil{2 \eps^{-1} m}$ samples.
			\State \textbf{choose} $j_1,\ldots,j_s \in [n]$ uniformly and independently at random.
			\State \textbf{let} $M$ be an uninitialized $m \times n$ sparse matrix $\{0,1\}$. \Comment{storage for speculative queries}
			\State \textbf{let} $A$ be an empty list over $[n]$.
			\State $c \gets 0$.
			\ForRange{$k$}{$1$}{$s$}
			\State \textbf{Invariant} $M_{i,j}$ is initialized for all $1 \le i \le c$ and $j \in \{j_1,\ldots,j_s\}$.
			\ForAll{$j$ \textbf{in} $A$}
			\Comment{simulation of $y^k$}
			\State \textbf{query} sample $k$ at $j$, giving $x^k_j$.
			\EndFor
			
			\State \textbf{set} found $\gets 0$.
			
			\ForRange{$i$}{$1$}{$c$}
			\If{$\bigwedge_{j \in A} { \left(M_{i,j} = x^k_j\right) }$}\Comment{simulation of the $y^i$s}
			\State \textbf{set} found $\gets 1$.
			\State $j \gets j_k$.
			\State \textbf{query} sample $k$ at $j$, giving $x^k_j$.
			\If{$M_{i,j} \ne x^k_j$}
			\State $c \gets c + 1$.
			\State \textbf{add} $j$ \textbf{to} $A$.
			\State \textbf{query} sample $k$ at $j_1,\ldots,j_s$, giving $M_{c,j_1},\ldots,M_{c,j_s}$. \Comment{speculative queries}
			\State \Comment{keep the invariant}
			\EndIf
			\EndIf
			\EndFor
			\If{found $= 0$}
			\State $c \gets c + 1$.
			\State \textbf{query} sample $k$ at $j_1,\ldots,j_s$, giving $M_{c,j_1},\ldots,M_{c,j_s}$. \Comment{speculative queries}
			\State \Comment{keep the invariant}
			\EndIf
			
			\If{$c > m$}
			\State \Return \reject
			\EndIf
			\EndFor
			\State \Return \accept
		\end{algorithmic}
	\end{algorithm}
	
	\subsection{Exponential separation from the locally bounded model} \label{sec:forward:subsec:exp-sep-local}
	
	As observed above, one of the advantages of forward-only algorithms over locally bounded ones is the ability to consider binary relations between samples. We show a logarithmic, forward-only $\eps$-test for $\mathbf{Inv}$, demonstrating the exponential separation between the models.
	
	\subsubsection*{Logarithmic forward-only $\eps$-testing algorithm for $\mathbf{Inv}$}
	
	We show an algorithm that $\eps$-tests $\mathbf{Inv}$ using $O(\eps^{-2})$ element queries (that translate to $O(\eps^{-2} \log n)$ bit queries in the binary encoding). We consider the first sample as ``$f$'', and observe that it allows us some indirect access to a presumptive $f^{-1}$ (even if it does not even exist, which is the case if $f$ is not a permutation). Then we take samples and try to distinguish them from both $f$ and the presumptive $f^{-1}$. If we manage to do it, we reject. After $\ceil{3 \eps^{-2}}$ tries, the probability to reject an $\eps$-far input is higher than $\frac{1}{2}$.
	
	\begin{algorithm}
		\caption{One sided $\eps$-test for $\mathbf{Inv}$, forward only, $O(\eps^{-2})$ queries}
		\label{alg:forward:p-inv}
		\begin{algorithmic}
			\State Treat samples as $n$-long strings over $[n]$.
			
			\State \textbf{let} $s = 1 + \ceil{3 \eps^{-2}}$.
			\State \textbf{choose} $j_2,\ldots,j_s \in [n]$, uniformly at random and independently.
			\State \textbf{choose} $k_2,\ldots,k_s \in [n]$, uniformly at random and independently.
			\State \textbf{query} sample 1 at $j_2,\ldots,j_s$, giving $f(j_2),\ldots,f(j_s)$.
			\State \textbf{query} sample 1 at $k_2,\ldots,k_s$, giving $f\left(k_2\right),\ldots,f\left(k_s\right)$.
			\ForRange{$i$}{$2$}{$s$}
			\State \textbf{query} sample $i$ at $j_i$, $f(k_i)$, giving $g(j_i)$, $g(f(k_i))$.
			\If{$f(j_i) \ne g(j_i)$ \textbf{and} $g(f(k_i)) \ne k_i$}
			\State \Return \reject
			\EndIf
			\EndFor
			\State \Return \accept
		\end{algorithmic}
	\end{algorithm}
	
	\begin{theorem} \label{th:alg-Inv-correct}
		Algorithm \ref{alg:forward:p-inv} is a one-sided $\eps$-test for $\mathbf{Inv}$.
	\end{theorem}
	
	\begin{proof}
		The query complexity is trivially $O(\eps^{-2})$.
		
		For perfect completeness, let $P \in \mathbf{Inv}$ be an input distribution that is supported either by some $\{f_0\}$ or by some $\{f_0,f_0^{-1}\}$. In the first case, $f(j_i) = g(j_i)$ for every $i$ and thus the algorithm cannot reject. In the second case, without loss of generality, assume that the first sample is $f_0$ (the analysis for $f_0^{-1}$ is the same). For every $i \ge 1$, and $g$ being the $i$th sample, either $g=f_0$, and then $f(j_i) = g(j_i)$, or $g=f_0^{-1}$, and then $g(f(k_i)) = k_i$. In both subcases, the algorithm cannot reject.
		
		For soundness, consider $P$ that is $\eps$-far from $\mathbf{Inv}$, and fix some function $f : [n] \to [n]$. Let $g$ be drawn from $P$. To reject, the algorithm seeks for two witnesses $j$ and $k$ such that $g(i) \ne i$ and $g(f(k)) \ne k$. For every specific $g$, the probability to reject is exactly $d(g,f) \cdot d(g \circ f, \mathrm{id})$.
		
		If $f$ is a permutation then $d(g\circ f,\mathrm{id}) = d(g,f^{-1})$, and in this case, by positivity of variance, we obtain for $g$ that is distributed as a random sample from $P$:
		\begin{eqnarray*}
			\E\left[\left(\min\left\{d\left(g,f\right),d\left(g,f^{-1}\right)\right\}\right)^2\right]
			&\ge& \left(\E\left[\min\left\{d\left(g,f\right),d\left(g,f^{-1}\right)\right\}\right]\right)^2 \\
			&=& \left(\E\left[d\left(g, \{f,f^{-1}\}\right)\right]\right)^2 \ge \eps^2
		\end{eqnarray*}
		
		If $f$ is $\frac{1}{3}\eps$-far from a permutation, then $d(g \circ f, \mathrm{id}) > \frac{1}{3}\eps$ for every $g$ as well, hence:
		\[\E\left[d\left(g,f\right) \cdot d\left(g \circ f, \mathrm{id}\right)\right] \ge \E\left[d\left(g,f\right) \cdot \frac{1}{3}\eps\right] \ge \frac{1}{3}\eps^2\]
		
		If $f$ is $\frac{1}{3}\eps$-close to a permutation, let $\tilde{f}$ be one of the closest permutations to $f$. For every specific $g$ it holds that $d(g \circ f, \mathrm{id}) \ge d(g \circ \tilde{f}, \mathrm{id}) - d(g \circ \tilde{f}, g \circ f) \ge d(g \circ \tilde{f}, \mathrm{id}) - d(\tilde{f}, f)$. Note that $\tilde{f}$ is necessarily constructed considering every set of preimages $f^{-1}(k)=\{i:f(i)=k\}$, and changing exactly $|f^{-1}(k)|-1$ values of $f$ in it. Considering any function $g$, the distance of $g \circ f$ from the identity, and even from being one-to-one, is at least $d(f, \tilde{f})$, hence 
		\begin{eqnarray*}
			d(g\circ f,\mathrm{id})
			&\ge& \max\{d(g\circ \tilde{f},\mathrm{id})-d(g\circ \tilde{f},g\circ f),d(f,\tilde{f})\} \\
			&\ge& \max\{d(g,\tilde{f}^{-1})-d(f,\tilde{f}),d(f,\tilde{f})\}
			\ge \frac12d(g,\tilde{f}^{-1})
		\end{eqnarray*}
		Hence,
		\begin{eqnarray*}
			& \E\left[d(g,f)\cdot d(g\circ f,\mathrm{id})\right] &
			\ge \E\left[(d(g,\tilde{f}) - \frac{1}{3}\eps) \cdot \frac{1}{2} d(g,\tilde{f}^{-1})\right] \\
			&\ge& \hspace{-48pt} \frac{1}{2} \E\left[\left(\min\{d(g,\tilde{f}),d(g,\tilde{f}^{-1})\} - \frac{1}{3}\eps\right) \min\{d(g,\tilde{f}),d(g,\tilde{f}^{-1})\}\right] \\
			&\underset{(*)}\ge& \hspace{-48pt} \frac{1}{2} \left(\E\left[\min\{d(g,\tilde{f}),d(g,\tilde{f}^{-1})\}\right] - \frac{1}{3}\eps\right)\left(\E\left[\min\{d(g,\tilde{f}),d(g,\tilde{f}^{-1})\}\right]\right) \\
			&\ge& \hspace{-48pt} \frac{1}{2}(\eps - \frac{1}{3}\eps)\eps = \frac{1}{3}\eps^2
		\end{eqnarray*}
		
		The starred transition is by a specific case of Jensen's inequality: $\E[(X-t)X] \ge (\E[X])^2 - t \E[X]$. For every specific choice of $f$ as the first sample, the probability to reject is at least $\frac{1}{3}\eps^2$, and thus after $\ceil{3 \eps^{-2}}$ iterations the probability to reject the input is higher than $\frac{1}{2}$.
		
		This completes the proof.
	\end{proof}
	
	\begin{corollary}
		Based on Theorem \ref{th:alg-Inv-correct}, Theorem \ref{th:lbnd-local-Inv} and the discussion in Subsection \ref{subsec:std-mdl-rdc}, there exists a property $\mathbf{Inv}^*$ of distributions over binary strings for which for every $\eps > 0$, there exists a forward-only $\eps$-tester for $\mathbf{Inv}^*$ that uses $O(\eps^{-2} \log n)$ bit queries, but every locally-bounded $\frac{1}{5}$-test must use at least $\Omega(\sqrt{n / \log n})$ bit queries.
	\end{corollary}
	
	Note that the string length of $\mathbf{Inv}^*$ is actually $O(n \log n)$ of the analysis for $\mathbf{Inv}$, hence the division in the lower bound.
	
	\subsection{Polynomial lower bound for $\mathbf{Sym}$} \label{sec:forward:subsec:poly-lbnd-Sym}
	
	\begin{theorem} \label{th:fwd-lbnd-sym}
		Every forward-only $\frac{1}{14}$-test for $\mathbf{Sym}$ must use at least $\frac{1}{2}\sqrt{m}$ queries (for sufficiently large $m$).
	\end{theorem}
	
	The whole subsection is dedicated to proving this theorem. We start with some definitions that will help us describe our test. First, we denote the set $[m]$ by $S$, and also refer it as the ``set of keys''.
	
	\begin{definition}[Key of an element for $\mathbf{Sym}$] \label{def:Sym:key}
		Let $x \in \{0,1\}^{2m}$. We define its \emph{key}, $\kappa(x)$, as the element in $S$ that is deduced from the first $\ceil{\log_2 m}$ bits of $x$.
	\end{definition}
	
	\begin{definition}[Valid key] \label{def:Sym:valid-key}
		A string $x \in \{0,1\}^{2m}$ has a \emph{valid key} if $\braket{x_1,\ldots,x_m} = C(\kappa(x))$ (recall that $C$ is a large-distance systematic code).
	\end{definition}
	
	\begin{definition}[Probability to draw a key] \label{def:Sym:pr-key}
		Let $P$ be a distribution over $\{0,1\}^{2m}$. For every $a \in S$, we define its \emph{probability to be drawn from $P$} as $\Pr_{x \sim P}[\kappa(x) = a]$, and denote it by $\Pr_P[a]$.
	\end{definition}
	
	\begin{definition}[Key support of $P$] \label{def:Sym:key-supp}
		For every $a \in S$, we say that the key $a$ \emph{appears in the support of $P$} if $\Pr_P[a] > 0$.
	\end{definition}
	
	\begin{definition}[Data of an element for $\mathbf{Sym}$] \label{def:Sym:data}
		Let $x \in \{0,1\}^{2m}$. The last $m$ bits (``right hand half'') of $x$ have one-to-one correspondence for the elements in $S$. For every $b \in S$ we define $\phi_x(b)$ as \emph{the value of $x$ in $b$}, corresponding to that one-to-one match. $\phi_x(b)$ can be explicitly defined as $x_{m+b}$, the $m+b$th bit of $x$.
	\end{definition}
	
	\begin{definition}[Consistency at $(a,b)$] \label{def:Sym:local-consistency}
		Let $P$ be a distribution over $\{0,1\}^{2m}$, and let $a, b \in S$. $P$ is \emph{consistent at $(a,b)$} if $\Pr_P[a] = 0$, or $\Pr_P[b] = 0$, or $\phi_x(b)$ is either $0$ with probability $1$ or $1$ with probability $1$, when $x$ is a random sample with $\kappa(x)=a$.
	\end{definition}
	
	\begin{definition}[Consistency] \label{def:Sym:consistency}
		Let $P$ be a distribution over $\{0,1\}^{2m}$. $P$ is \emph{consistent} if it is consistent at every $(a,b)$ ($a\ne b \in S$).
	\end{definition}
	
	\begin{definition}[Symmetry at $\{a,b\}$] \label{def:Sym:local-symmetry}
		Let $P$ be a distribution over $\{0,1\}^{2m}$, and let $a,b \in S$. $P$ is \emph{symmetric at $\{a,b\}$} if
		$\Pr_{x,y \sim P}\left[\kappa(x) = a \wedge \kappa(y) = b \wedge \phi_x(b) \ne \phi_y(a) \right] = 0$.
	\end{definition}
	
	\begin{definition}[Symmetry] \label{def:Sym:symmetry}
		Let $P$ be a distribution over $\{0,1\}^{2m}$. $P$ is \emph{symmetric} if it is symmetric at every $\{a,b\}$ (where $a \ne b \in S$).
	\end{definition}
	
	\begin{observation} \label{obs:Sym-iff-symmetric}
		$P \in \mathbf{Sym}$ if and only if it is symmetric.
	\end{observation}
	
	\begin{observation}
		If $P$ is symmetric at $\{a,b\}$ then it is also consistent at both $(a,b)$ and $(b,a)$.
	\end{observation}
	
	\begin{observation} \label{obs:Sym:subset-of-support-ok}
		If $P \in \mathbf{Sym}$, then for every distribution $Q$, if $\supp(Q) \subseteq \supp(P)$ then $Q \in \mathbf{Sym}$.
	\end{observation}
	\begin{proof}
		The definition of $\mathbf{Sym}$ has the following form: ``For every $x,y\in\mathrm{supp}(P)$, if $\kappa(x)\neq\kappa(y)$ then they satisfy the symmetry condition''. If $\supp(Q) \subseteq \supp(P)$, and $P \in \mathbf{Sym}$, then for every $x,y \in \supp(Q)$, they belong to $\supp(P)$ as well and thus they still satisfy the condition above.
	\end{proof}
	
	To proceed with the proof of Theorem \ref{th:fwd-lbnd-sym}, we define a useful construction of distributions.
	
	\begin{definition}[Distribution $U_f$] \label{def:uf-for-Sym}
		Let $f : S^2 \to \{0,1\}$ be a function. We define $U_f$ as the uniform distribution over $\left\{C(a)\braket{f(a,b)|b\in S} \cond a \in S\right\} \subseteq \{0,1\}^{2m}$.
	\end{definition}
	
	\begin{observation} \label{obs:uf-Sym-iff-f-sym}
		$U_f \in \mathbf{Sym}$ if and only $f \in \mathbf{sym}$.
	\end{observation}
	
	\begin{lemma}
		\label{lemma:uf-sym-f-sym-distance}
		For every $f : S^2 \to \{0,1\}$, it holds that $d(U_f,\mathbf{Sym}) \ge \frac{1}{6}d(f, \mathbf{sym})$.
	\end{lemma}
	The proof of this lemma appears at the end of this section. Based on the ``standard model'' distance bound that it guarantees, we define two distributions over inputs:
	\begin{itemize}
		\item $D_\mathrm{yes}$ chooses uniformly at random a symmetric function $f : S^2 \to \{0,1\}$, and then outputs $U_{f}$.
		\item $D_\mathrm{no}$ chooses uniformly at random an anti-symmetric function $f : S^2 \to \{0,1\}$, and then outputs $U_{f}$. By ``anti-symmetric'' we mean that $f(a,b) \oplus f(b,a) = 1$ for every $a \ne b \in S$.
	\end{itemize}
	
	Observe that $D_\mathrm{yes}$ draws an input in $\mathbf{Sym}$ with probability $1$. The $f$ that is drawn by $D_\mathrm{no}$ is always $\binom{m}{2}/m^2$-far from $\mathbf{sym}$, which is $\frac{1}{2}-o(1)$. Hence by Lemma \ref{lemma:uf-sym-f-sym-distance}, an input that is drawn from $D_\mathrm{no}$ is $\frac{1}{12}-o(1)$-far from $\mathbf{Sym}$.
	
	\begin{lemma}[No useful queries lemma]
		\label{lemma:no-useful-queries}
		Let $f : S^2 \to \{0,1\}$ be a function, and let $\mathcal{A}$ be a forward-only probabilistic algorithm that uses $s$ samples and $q$ queries. If the input has the form of $U_f$ for some $f:S^2\to\{0,1\}$, then with probability higher than $1 - \frac{sq}{m}$, for every $a \ne b \in S$, the algorithm obtains at most one of the values $f(a,b)$ or $f(b,a)$.
	\end{lemma}
	
	\begin{proof}
		For every $1 \le i \le q$, let $X_i$ be an indicator for the following event: there exist $i' > i$ and $a \ne b \in S$ such that the $i$th query obtains $f(a,b)$, and the $i'$th query obtains $f(b,a)$.
		
		Fix some $i$, and assume that the $i$th query obtains $f(a_i,b_i)$ for some $a_i \ne b_i$. The $i$th query is made in some $j(i)$th sample whose key is $a_i$. To be able to obtain $f(b_i,a_i)$, there must be a sample whose index is $j > j(i)$ and whose key is $b_i$. The $j$th sample (for every $j > j(i)$) is completely independent of the algorithm's behavior so far, because the algorithm is forward-only and has never had any interaction with this sample. Hence the probability that its key is $b_i$ is $\frac{1}{n}$, and by the union bound, $\Pr[X_i = 1] \le (s - j(i))\cdot \frac{1}{n} \le \frac{s}{n}$. Considering all $q$ queries, by the union bound, $\Pr[\exists i : X_i = 1] \le \frac{sq}{n}$.
	\end{proof}
	
	Now we can complete the proof of Theorem \ref{th:fwd-lbnd-sym}.
	
	\begin{proof}[Proof (of Theorem \ref{th:fwd-lbnd-sym})]
		Consider a probabilistic forward-only algorithm $\mathcal{A}$ that makes less than $q \le \frac{1}{2}\sqrt{m}$ queries, and without loss of generality, at most $q$ samples ($s \le q$). By Lemma \ref{lemma:no-useful-queries}, if the algorithm is executed on an input $U_f$, then with probability at least $1 - \frac{1}{4}$ there are no $a \ne b \in S$ for which the algorithm gathers both $f(a,b)$ and $f(b,a)$. For both $D_\mathrm{yes}$ and $D_\mathrm{no}$, the distribution of answers is the same (completely uniform for the queries taken from the data part of the samples). Thus, the variation distance between the algorithm's behavior on $D_\mathrm{yes}$ and $D_\mathrm{no}$ is at most $\frac{1}{4}$. Hence by Yao's principle, the algorithm cannot be a $\frac{1}{5}$-test for $\mathbf{inv}$.
	\end{proof}
	
	Finally we present the proof of Lemma \ref{lemma:uf-sym-f-sym-distance} that was postponed earlier.
	
	\begin{proof}[Proof (of Lemma \ref{lemma:uf-sym-f-sym-distance})]
		Let $h : \supp(P) \to \{0,1\}^{2m}$ be the mapping that is guaranteed by Lemma \ref{lemma:canonical-dist-from-support-property} (which is applicable due to Observation \ref{obs:Sym:subset-of-support-ok}), that is, $h(P) \in \mathbf{Sym}$ and $d(P,h(P)) = d(P,\mathbf{Sym})$. For every $a \in S$, let $x^a$ be the only element in the support of $U_f$ whose key is $a$.
		
		Let $g$ be a symmetric function that is made by fixing all violations in $f$ using ``hints'' from $h$. Formally,
		\begin{align*}
			g(a,b) = \begin{cases}
				\phi_{h(x^a)}(b) & \kappa(h(x^a)) = a \\
				\phi_{h(x^b)}(a) & \kappa(h(x^a)) \ne a, \kappa(h(x^b)) = b \\
				0 & \kappa(h(x^a)) \ne a, \kappa(h(x^b)) \ne b
			\end{cases}
		\end{align*}
		Observe that $g$ is symmetric, and let $h' : \supp\left(U_f\right) \to \{0,1\}^{2m}$ be the following mapping: \[h'(x) = \left\langle x_1,\ldots,x_m \right\rangle \braket{ g(\kappa(x),b) | b \in S }\]
		
		Observe that $d(x^a, h'(x^a)) \le \frac{1}{2}$ for every $a \in S$ (because their key parts match) and that if $\kappa(h(x^a)) \ne a$ then $d(x^a,h(x^a)) \ge \frac{1}{6}$, because codewords for different keys are $\frac{1}{3}$-far apart, and the weight of the key is $\frac{1}{2}$.
		
		For every $a \in S$: if $\kappa(h(x^a)) = a$, then $h(x^a) = h'(x^a)$, hence $d(x^a,h(x^a)) \ge \frac{1}{3} d(x^a,h'(x^a))$. Otherwise, $d(x^a,h(x^a)) \ge \frac{1}{6} \ge \frac{1}{3} d(x^a,h'(x^a))$ as well. In total,
		\begin{eqnarray*}
			d(U_f, \mathbf{Sym})
			= d(U_f, h(U_f))
			&=& \sum_{a \in S}{\frac{1}{m} d(x^a, h(x^a))} \\
			&\ge& \frac{1}{3} \sum_{a \in S}{\frac{1}{m} d(x^a, h'(x^a))}
			= \frac{1}{6} d(f, g)
			\ge \frac{1}{6} d(f, \mathbf{sym})
		\end{eqnarray*}
	\end{proof}
	
	\section{The constant memory model} \label{sec:kmem}
	
	In this section we discuss some characteristics of bounded memory models. The intuition is that $k$-memory algorithms can handle $k$-ary relationships of elements, while smaller memories cannot do that. We prove this intuition by separating weak $k$-memories from strong $k-1$ ones.
	
	\subsection{Exponential separation of forward-only and weak $2$-memory bounded} \label{sec:kmem:subsec:exp-sep-forward-bounded}
	
	In this section we show that $\mathbf{Sym}$ is $\eps$-testable using $O\left(\poly\left(\eps^{-1}\right) \log{n}\right)$ queries by a weak $2$-memory adaptive algorithm, hence demonstrating an exponential separation from the forward-only model.
	
	\subsubsection*{Logarithmic, weak $2$-memory adaptive $\eps$-testing algorithm}
	
	The $\eps$-test for $\mathbf{Sym}$ is straightforward: it uses sufficiently many iterations (in particular, $\ceil{8 \eps^{-2}}$), each one of them consisting of taking two samples and validating their keys and symmetry (with respect to their keys). The bottleneck of the query complexity is actually reading the key of every sample, which is logarithmic, rather than the validation itself, which requires exactly four queries per iteration (two of them to validate the keys, and two more to validate symmetry).
	
	\begin{algorithm}[H]
		\caption{One-sided $\eps$-test for $\mathbf{Sym}$, weak $2$-memory, $O(\eps^{-2} \log n)$ queries}
		\label{alg:w-2-mem:psym}
		\begin{algorithmic}
			\State \textbf{let} $m \gets n/2$.
			\For{$\ceil{8 \eps^{-2}}$ \textbf{times}}
			\State take two samples $x$, $y$.
			\State \textbf{query} $x_1,\ldots,x_{\ceil{\log_2 m}}$, giving $\kappa(x)$ as $a$.
			\State \textbf{query} $y_1,\ldots,y_{\ceil{\log_2 m}}$, giving $\kappa(y)$ as $b$.
			\State \textbf{choose} $i \in [m]$, uniformly at random.
			\State \textbf{query} $x$, $y$ at $i$, giving $x_i$, $y_i$.
			\State \textbf{query} $\phi_x(b)$, $\phi_y(a)$.
			\If{$x_i \ne (C(a))_i$ \textbf{or} $y_i \ne (C(b))_i$}
			\State \Return \reject \Comment{rejection by key invalidity}
			\EndIf
			\If{$\phi_x(b) \ne \phi_y(a)$}
			\State \Return \reject \Comment{rejection by asymmetry}
			\EndIf
			\EndFor
			\State \Return \accept
		\end{algorithmic}
	\end{algorithm}
	
	To be able to analyze the upper bound for $\mathbf{Sym}$, we need some additional definitions.
	
	\begin{definition}[$p_{a,b}$, ``zeroness'' of the presumed $f(a,b)$] \label{def:Sym:pab}
		Let $a,b \in S$ for which $\Pr_P[a] > 0$. We set $p_{a,b} = \Pr_{x \sim P}\left[\phi_x(b) = 0 \cond  \kappa(x) = a\right]$.
	\end{definition}
	
	\begin{definition}[Specific fixing cost, $c_{a,b,x}$] \label{def:Sym:specific-fixing-cost-abx}
		Let $P$ be a distribution over $\{0,1\}^{2m}$. For $a,b \in S$ for which $\Pr_P[a], \Pr_P[b] > 0$, let the zero-fix cost be $c_{a,b,0} = \frac{1}{2m} ((1 - p_{a,b})\Pr_P[a] + (1 - p_{b,a})\Pr_P[b])$, and the one-fix cost be $c_{a,b,1} = \frac{1}{2m} (p_{a,b} \Pr_P[a] + p_{b,a} \Pr_P[b])$. In other words, for $x \in \{0,1\}$, $c_{a,b,x}$ is the cost of making $P$ symmetric at $\{a,b\}$ where both values are $x$.
	\end{definition}
	
	\begin{definition}[Fixing cost, $c_{a,b}$] \label{def:Sym:fixing-cost-ab}
		Let $P$ be a distribution over $\{0,1\}^{2m}$. For $a,b \in S$ for which $\Pr_P[a], \Pr_P[b] > 0$, let the fixing cost be $c_{a,b} = \min\{c_{a,b,0}, c_{a,b,1}\}$. In other words, $c_{a,b}$ is the earth mover's cost of making $P$ symmetric at $(a,b)$.
	\end{definition}
	
	\begin{observation}
		For every $a,b \in S$, $c_{a,b,0} = c_{b,a,0}$ and $c_{a,b,1} = c_{b,a,1}$.
	\end{observation}
	
	\begin{lemma} \label{lemma:Sym-fixing-cost-bound}
		For $a,b \in S$ for which $\Pr_P[a], \Pr_P[b] > 0$,
		\[c_{a,b} \le \frac{\Pr_P[a] + \Pr_P[b] }{2 m \Pr_P[a]\Pr_P[b]} \Pr\left[\kappa(x) = a \wedge \kappa(y) = b \wedge \phi_x(b) \ne \phi_y(a)\right]\]
	\end{lemma}
	\begin{proof}
		Let $\rho = \Pr\left[\phi_x(b) \ne \phi_y(a) \cond \kappa(x) = a \wedge \kappa(y) = b\right] = (1-p_{a,b})p_{b,a} + p_{a,b}(1-p_{b,a})$.
		
		\paragraph{Case I. $p_{a,b} \le \frac{1}{2} \le p_{b,a}$} Observe that $\rho \ge \frac{1}{2}$ in this case, hence
		\begin{eqnarray*}
			c_{a,b}
			&=& \frac{1}{2m}\min\left\{(1-p_{a,b})\Pr_P[a] + (1-p_{b,a})\Pr_P[b],\  p_{a,b}\Pr_P[a] + p_{b,a}\Pr_P[b]\right\} \\
			&\le& \frac{1}{2m} \cdot \frac{1}{2} \left(\Pr_P[a] + \Pr_P[b]\right) \\
			&\le& \frac{\Pr_P[a] + \Pr_P[b]}{2m}\rho \\
			&=& \frac{\Pr_P[a] + \Pr_P[b]}{2m \Pr_P[a] \Pr_P[b]}\Pr\left[\kappa(x) = a \wedge \kappa(y) = b \wedge \phi_x(b) \ne \phi_y(a)\right]
		\end{eqnarray*}
		
		\paragraph{Case II. $p_{a,b}, p_{b,a} \le \frac{1}{2}$} Observe that $\rho \ge \max\{p_{a,b},p_{b,a}\}$ in this case, because
		$(1-p_{a,b})p_{b,a} + p_{a,b}(1-p_{b,a}) = p_{a,b} + (1 - 2 p_{a,b}) p_{b,a} \ge p_{a,b}$ (and $\rho \ge p_{b,a}$ analogously), hence
		\begin{eqnarray*}
			c_{a,b}
			&=& \frac{1}{2m}\min\left\{(1-p_{a,b})\Pr_P[a] + (1-p_{b,a})\Pr_P[b],\  p_{a,b}\Pr_P[a] + p_{b,a}\Pr_P[b]\right\} \\
			&=& \frac{1}{2m} \left(p_{a,b}\Pr_P[a] + p_{b,a}\Pr_P[b]\right) \\
			&\le& \frac{\Pr_P[a] + \Pr_P[b]}{2m} \max\{p_{a,b},p_{b,a}\} \\
			&\le& \frac{\Pr_P[a] + \Pr_P[b]}{2m}\rho \\
			&=& \frac{\Pr_P[a] + \Pr_P[b]}{2m \Pr_P[a] \Pr_P[b]}\Pr\left[\kappa(x) = a \wedge \kappa(y) = b \wedge \phi_x(b) \ne \phi_y(a)\right]
		\end{eqnarray*}
		
		The case where $p_{a,b},p_{b,a}\ge\frac{1}{2}$ is handled similarly to Case I by replacing $p_{a,b}$ and $p_{b,a}$ with $1-p_{a,b}$ and $1-p_{b,a}$ respectively. Analogously, the remaining case, $p_{b,a}\le\frac{1}{2}\le p_{a,b}$, can be handled similarly to Case II.
	\end{proof}
	
	\begin{definition}[Key invalidity] \label{def:Sym:key-invalidity}
		For an input distribution $P$, we define its \emph{key invalidity} as:
		\[K(P) = \E_{x \sim P}\left[ d(\braket{x_1,\ldots,x_m},C(\kappa(x))) \right] = \E_{x \sim P, i \sim [m]}\left[ x_i \ne (C(\kappa(x)))_i\right]\]
	\end{definition}
	Key invalidity is a measure for ``how far is $P$ from having valid keys'', and it is also the probability of a single iteration to reject a sample $x$ by key invalidity.
	
	\begin{definition}[Asymmetry] \label{def:Sym:asymmetry}
		For an input distribution $P$, we define its \emph{asymmetry} as:
		\[I(P) = \E_{x,y \sim P}\left[\phi_x(\kappa(y)) \ne \phi_y(\kappa(x))\right]\]
	\end{definition}
	Asymmetry is a measure for ``how far is $P$ from being symmetric'', and also the probability of the algorithm to reject by asymmetry.
	
	\begin{observation}
		The probability to reject an input $P$ is at least $\max\{K(P), I(P)\}$.
	\end{observation}
	\begin{proof}
		Immediately, by the definitions of $K(P)$ and $I(P)$.
	\end{proof}
	
	\begin{theorem} \label{th:alg-psym-correct}
		Algorithm \ref{alg:w-2-mem:psym} is a one-sided $\eps$-test of $\mathbf{Sym}$.
	\end{theorem}
	
	\begin{proof}
		The complexity of a single iteration is two samples and $O(\log n)$ queries. In total, the algorithm uses $O(\eps^{-2})$ samples and $O(\eps^{-2} \log n)$ queries.
		
		For perfect completeness, consider some $P \in \mathbf{Sym}$. It must be supported by a set of elements with valid keys such that each two of them do not violate symmetry.
		
		For soundness, consider some input distribution $P$ that is $\eps$-far from $\mathbf{Sym}$. Let $\delta = \frac{1}{2}\eps$ and let $\tilde{S} \subseteq S$ be the set of keys whose probability in $P$ is at least $\frac{\delta}{m}$. Let $f : \tilde{S}^2 \to \{0,1\}$ be the following function:
		\begin{align*}
			f(a,b) = \begin{cases}
				1 & c_{a,b,1} \le c_{a,b,0} \\
				0 & \mathrm{otherwise}
			\end{cases}
		\end{align*}
		Observe that $f$ is symmetric, because $a$ and $b$ have the exact same role in its definition. Let $h : \{0,1\}^{2m} \to \{0,1\}^{2m}$ be the following map:
		\begin{align*}
			h(x) = C(\kappa(x))\braket{ \begin{cases}
					f(\kappa(x),b) & \kappa(x),b \in \tilde{S} \\
					\phi_x(b) & \mathrm{otherwise}
				\end{cases} \cond b \in S}
		\end{align*}
		
		The distribution $h(P)$ does not necessarily belong to $\mathbf{Sym}$, but it is $\delta$-close to it: if we delete all samples whose key is not in $\tilde{S}$ (and transfer their probabilities arbitrarily), the resulting distribution does belong to $\mathbf{Sym}$. Below we bound the distance of $P$ from $h(P)$.
		\begin{eqnarray*}
			& d(P,h(P)) &  \le \quad \sum_{a \in S} \Pr_P[a] \E_{x \sim P}\left[d(x,h(x)) \cond \kappa(x) = a\right] \\
			&\hspace{-24pt} \le & \hspace{-36pt} \frac{1}{2} K(P) + \frac{1}{2} \sum_{a \in \tilde{S}} \Pr_P[a] \E_{x \sim P}\left[d(\braket{x_{m+1},\ldots,x_{2m}},\braket{h_{m+1}(x),\ldots,h_{2m}(x)}) \cond \kappa(x) = a\right] \\
			&\hspace{-24pt}  \le\ & \hspace{-36pt} \frac{1}{2}K(P) +\sum_{a,b \in \tilde{S}} c_{a,b}\\
			&\hspace{-24pt}  \underset{(*)}\le\ & \hspace{-36pt} \frac{1}{2}K(P) + \sum_{a,b \in \tilde{S}} \frac{\Pr_P[a] + \Pr_P[b]}{2 m \Pr_P[a] \Pr_P[b]} \Pr_{x,y \sim P}\left[\kappa(x) = a \wedge \kappa(y) = b \wedge \phi_x(b) \ne \phi_y(a) \right]\\
			&\hspace{-24pt} \underset{(**)}\le\ & \hspace{-36pt} \frac{1}{2}K(P) \!+\!\! \sum_{a,b \in \tilde{S}} \frac{2 \max\{\Pr_P[a], \Pr_P[b]\}}{2 m \frac{\delta}{m} \max\{\Pr_P[a], \Pr_P[b]\}} \Pr_{x,y \sim P}\left[\kappa(x) = a \wedge \kappa(y) = b \wedge \phi_x(b) \ne \phi_y(a) \right]\\
			&\hspace{-24pt}  =\ & \hspace{-36pt} \frac{1}{2}K(P) + \delta^{-1} \sum_{a,b \in \tilde{S}} \Pr_{x,y \sim P}\left[\kappa(x) = a \wedge \kappa(y) = b \wedge \phi_x(b) \ne \phi_y(a) \right] \\
			&\hspace{-24pt} \le\ & \hspace{-36pt} \frac{1}{2}K(P) + \delta^{-1} I(P) .\\
		\end{eqnarray*}
		
		The first transition is correct because we can use a transfer distribution that maps every $x$ to its $h(x)$, and the rightmost sum only considers keys in $\tilde{S}$ because $h$ does not modify values of samples with rare keys. The starred transition is by Lemma \ref{lemma:Sym-fixing-cost-bound}, and the doubly-starred transition is correct because $\Pr[a], \Pr[b] \ge \frac{\delta}{m}$. Other transitions are trivial. Using the above we bound the distance of $P$ from $\mathbf{Sym}$.
		\begin{align*}
			d(P,\mathbf{Sym}) &\ \le \ & \delta + d(P,h(P)) 
			&\ \le\ & \delta + \frac{1}{2}K(P) + \delta^{-1} I(P)
			&\ \le \ & \frac{1}{2}\eps + \frac{1}{2}K(P) + 2\eps^{-1} I(P) .
		\end{align*}
		
		Consider an input distribution $P$ that is $\eps$-far from $\mathbf{Sym}$. By the triangle inequality, either $K(P) > \frac{1}{2}\eps$ or $I(P) > \frac{1}{8}\eps^2$ (or both). In both cases, the probability to reject in a single iteration is at least $\frac{1}{8} \eps^2$. After $\ceil{8 \eps^{-1}}$ iterations, the probability to reject is greater than $\frac{1}{2}$.
	\end{proof}
	
	\subsection{Introduction to exponential separation of constant memories}
	
	Subsection \ref{sec:forward:subsec:poly-lbnd-Sym} contains a lower bound for $\mathbf{Sym}$, based on the concept of ``useful queries'' and the low probability to obtain them. We generalize it for $k$-set functions in order to prove stronger results. Note that here we ``lump together'' $k$ bit values for a set $\{a_1,\ldots,a_k\}$, while in the case for $\mathbf{Sym}$ we partition the two bits for a set $\{a,b\}\in\binom{S}{2}$ to $f(a,b)$ and $f(b,a)$.
	
	\begin{definition}[String properties \textsc{even}, \textsc{odd}]
		\textsc{even} is the property of binary strings with even parity. \textsc{odd} is the property of binary strings with odd parity.
	\end{definition}
	
	\begin{definition}[Function property $\mathbf{par}_k$, counterpart to Definition \ref{def:prop:sym}]
		Let $k \ge 2$. For a fixed $m$ and $S = [m]$, the property $\mathbf{par}_k$ is defined as the set of functions $f : \binom{S}{k} \to \{0,1\}^k$ such that for every $A \in \binom{S}{k}$, $f(k) \in \textsc{even}$.
	\end{definition}
	
	Our goal is to define a property $\mathbf{Par}_k$ of distributions that relates to $\mathbf{par}_k$ in the same way that $\mathbf{Sym}$ relates to $\mathbf{sym}$. To be more specific, we have the following informal constraints:
	
	\begin{enumerate}
		\item A weak $k$-memory algorithm can obtain ``a new value of $f$'' (with high probability) at the cost of $k$ samples and $O(k \log m)$ queries.
		\item For every $k' < k$ and for every strong $k'$-memory algorithm, the probability to obtain strictly more than $k'$ bits of even one value of $f$, is $O\left(\frac{ksq}{m}\right)$, where $s$ and $q$ are the number of samples and queries (respectively).
	\end{enumerate}
	
	\subsubsection*{The parity property}
	
	We generalize $\mathbf{Sym}$ to be able to describe functions from $\binom{S}{k}$ to $\{0,1\}^k$. Note that the generalization does not actually contain $\mathbf{Sym}$ itself, because the latter uses functions from $S^2$ to $\{0,1\}$, rather than from $\binom{S}{2}$ to $\{0,1\}^2$ (which we would achieve by ignoring the diagonal ``$f(a,a)$'' values and concatenating every $f(a,b)$, $f(b,a)$ into $f(\{a,b\})$ of length $2$), but other concepts are still relevant.
	
	The property is denoted by $\mathbf{Par}_k$. It has one explicit parameter $k$. For a size parameter $m$, the property is defined for distributions over $\{0,1\}^{2n}$, where $n = \binom{m-1}{k-1}$. Below we define the notions that we use in $\mathbf{Par}_k$.
	
	The following notions are identical to their counterparts in Subsection \ref{sec:kmem:subsec:exp-sep-forward-bounded}: $S = [m]$, key (Definition \ref{def:Sym:key}), valid key (Definition \ref{def:Sym:valid-key}), probability of a key (Definition \ref{def:Sym:pr-key}), key support (Definition \ref{def:Sym:key-supp}). In all of them, the length of the string is $2n$ (rather than $2m$), and the key part is $n$-bit long (rather than $m$).
	
	\begin{definition}[Data of an element for $\mathbf{Par}_k$, counterpart to Definition \ref{def:Sym:data}]
		Let $x \in \{0,1\}^{2n}$. As $n = \binom{m-1}{k-1}$, the last $n$ bits of $x$ have a correspondence to the subsets of $S \setminus \{\kappa(x)\}$ of size $k-1$. For every such set $A$ we define $\phi_x(A)$ as the value of $x$ in $A$. If $\kappa(x) \in A$ and also $\card{A} = k$, we define $\Phi_x(A)$ as $\phi_x(A \setminus \{\kappa(x)\})$.
	\end{definition}
	
	\begin{definition}[Consistency at $A$, counterpart to Definition \ref{def:Sym:local-consistency}] \label{def:p-pi-k:local-consistency}
		Let $P$ be a distribution over the set $\{0,1\}^{2n}$, and let $a_1 < \ldots < a_k \in S$ and $A = \{a_1,\ldots,a_k\}$. $P$ is \emph{consistent at $A$} if there exists a string $s \in \{0,1\}^k$ for which:
		\[ \Pr_{x^1,\ldots,x^k \sim P}\left[ \left(\bigwedge_{i=1}^{k} \left(\kappa(x^i) = a_i\right) \right) \wedge \braket{\Phi_{x^1}(A), \ldots, \Phi_{x^k}(A)} \ne s \right] = 0 \]
	\end{definition}
	
	\begin{definition}[Consistency, counterpart to Definition \ref{def:Sym:consistency}]
		Let $P$ be a distribution over $\{0,1\}^{2n}$. $P$ is \emph{consistent} if it is consistent at every $A \in \binom{S}{k}$.
	\end{definition}
	
	\begin{definition}[parity-validity at $A$, counterpart to Definition \ref{def:Sym:local-symmetry}]
		Let $P$ be a distribution over $\{0,1\}^{2n}$, and let $a_1 < \ldots < a_k \in S$ and $A = \{a_1,\ldots,a_k\}$. $P$ is \emph{parity-valid at $A$} if
		\[ \Pr_{x^1,\ldots,x^k \sim P}\left[\left(\bigwedge_{i=1}^{k} \kappa(x^i) = a_i\right) \wedge \bigoplus_{i=1}^k \Phi_{x^i}(A) = 1\right] = 0 \]
	\end{definition}
	
	\begin{definition}[parity-validity, counterpart to Definition \ref{def:Sym:symmetry}] \label{def:p-pi-k:pi-valid}
		Let $P$ be a distribution over the set $\{0,1\}^{2n}$. $P$ is \emph{parity-valid} if it is parity-valid at every $A \in \binom{S}{k}$.
	\end{definition}
	
	\begin{definition}[Property $\mathbf{Par}_k$]
		For a size parameter $m$, $n = \binom{m-1}{k-1}$ and a systematic code $C : [m] \to \{0,1\}^n$ (whose existence is guaranteed by Lemma \ref{lemma:syscode}), $\mathbf{Par}_k$ is the property of parity-valid distributions over $\{0,1\}^{2n}$ with valid keys. Note that if $P \in \mathbf{Par}_k$, then for every distribution $Q$, if $\supp(Q) \subseteq \supp(P)$ then $Q \in \mathbf{Par}_k$ (see Observation \ref{obs:Sym:subset-of-support-ok}).
	\end{definition}
	
	\begin{lemma}
		Let $P$ be a distribution. Parity-validity at $A$ implies consistency at $A$.
	\end{lemma}
	\begin{proof}
		Assume that $P$ is inconsistent at some $A = \{a_1,\ldots,a_k\}$ due to some key $a_i \in A$. That is, $a_1,\ldots,a_k$ appear in the support of $P$, and also
		\[\Pr_{y,y' \sim P}\left[\phi_y(A \setminus \{a_i\}) \ne \phi_{y'}(A \setminus \{a_i\}) \cond \kappa(y) = \kappa(y') = a_i \right] > 0\]
		Let $y,y'$ be two samples in the support of $P$ whose key is $a_i$, but they differ in their $\phi(A \setminus \{a_i\})$. Consider some choice $x_1,\ldots,x_{i-1},x_{i+1},\ldots,x_k$ of samples
		with positive probabilities for which $\kappa(x_j)=a_j$ for $1\le j\le
		k$, $j\ne i$, and consider the following two sequences:
		\begin{align*}
			x^1,\ldots,x^{i-1},y,x^{i+1},\ldots,x^k &&&
			x^1,\ldots,x^{i-1},y',x^{i+1},\ldots,x^k
		\end{align*}
		Both of these sequences have strictly positive probability to be chosen, and they represent two words that differ only in the $i$th bit. Hence, one of them does not belong to \textsc{even}.
	\end{proof}
	
	As in the analysis of $\mathbf{Sym}$, we define a useful construction of a distribution from a given function. Recall Definition \ref{def:ord-a-A} of $\mathrm{ord}(a,A)$ for $a\in A\in\binom{S}{k}$.
	
	\begin{definition}[$U_f$, counterpart to Definition \ref{def:uf-for-Sym}] \label{def:uf-for-p-pi-k}
		Let $f : \binom{S}{k} \to \{0,1\}^k$ be a function. We define $U_f$ as the following distribution: we first choose a key $a \sim S$ uniformly at random, and then return the following string:
		\[C(a) \braket{(f(B\cup\{a\}))_{\mathrm{ord}(a,B\cup\{a\})} \cond B \in \binom{S \setminus \{a\}}{k-1} }\]
	\end{definition}
	
	\begin{observation}[see Observation \ref{obs:uf-Sym-iff-f-sym}]
		$U_f \in \mathbf{Par}_k$ if and only $f \in \mathbf{par}_k$.
	\end{observation}
	
	\subsubsection*{Polynomial lower bound for $\mathbf{Par}_k$ for strong $k'$-memory where $k' < k$}
	
	To be able to show a polynomial lower bound we have to show two lemmas.
	
	\begin{lemma}[see Lemma \ref{lemma:uf-sym-f-sym-distance}]
		\label{lemma:uf-p-pi-k-f-pi-k-distance}
		For all $k \ge 2$ and $f : \binom{S}{k} \to \{0,1\}^k$, $ d(f, \mathbf{par}_k) \ge \frac{1}{6} d(U_f, \mathbf{Par}_k)$.
	\end{lemma}
	In the context of this lemma, the distance of two functions $f,g : \binom{S}{k} \to \{0,1\}^k$ is their Hamming distance as $k\binom{m}{k}$-bit strings, rather than Hamming distance as $\binom{m}{k}$-long strings over the alphabet $\{0,1\}^k$. Equivalently, $d(f,g)=\E_{x\sim\binom{S}{k}}d_{\mathrm{H}}(f(x),g(x))$.
	
	\begin{proof}
		Without loss of generality, let $h : \supp(P) \to \{0,1\}^{2n}$ be the mapping that is guaranteed by Lemma \ref{lemma:canonical-dist-from-support-property}, that is, $h(P) \in \mathbf{Par}_k$ and $d(P,\mathbf{Par}_k) = d(P,h(P)) = \sum_{a \in S}{\frac{1}{m} d(x^a, h(x^a))}$. For every $a \in S$, let $x^a$ be the only element in the support of $U_f$ whose key is $a$. Let $g : \binom{S}{k} \to \{0,1\}^k$ be the following function: for $A = \{a_1,\ldots,a_k\}$ where $a_1 < \ldots < a_k$, we have two cases.
		
		\paragraph{Case I} If, for every $1 \le i \le k$, there exists $b_i \in S$ such that $h(\kappa(x^{b_i})) = a_i$, then we define
		$g(A) = \braket{\phi_{h(x^{b_i})}(A \setminus \{a_i\}) \cond 1 \le i \le k}$. Note that for every choice of $b_1,\ldots,b_k$ such that $\kappa(h(x^{b_i})) = a_i$ for all $1 \le i \le k$, we have the exact same result of $\braket{\phi_{h(x^{b_i})}(A\setminus \{a_i\}) \cond 1 \le i \le k}$. This is due to the fact that all keys of $A$ appear in the support of $h(P)$, which belongs to $\mathbf{Par}_k$, and hence $h(P)$ is consistent at $A$. Also, this string must have even parity for the same reason. 
		
		\paragraph{Case II} Otherwise, we define $i_0$ as $i_0 = \min \left\{ 1 \le i \le k : \forall b \in S : h(\kappa(x^b)) \ne a_i \right\}$.
		
		\begin{align*}
			g(A) = \braket{
				\begin{cases}
					\phi_{h(x^{a_i})}(A \setminus \{a_i\}) & \kappa(h(x^{a_i})) = a_i \\
					\phi_{x^{a_i}}(A \setminus \{a_i\}) & \kappa(h(x^{a_i})) \ne a_i \wedge i \ne i_0 \\
					\bigoplus_{j=[k] \setminus \{i_0\}} (g(A))_j & i = i_0
				\end{cases}
				\cond 1 \le i \le k} .
		\end{align*}
		This makes sure that $g(A) \in \textsc{even}$ (and thus valid).
		
		Note that both in Case I and Case II, if $\kappa(h(x^{a_i})) = a_i$ for some $a_i \in A$ (not necessarily for other keys in $A$), then specifically $(g(A))_i = (f(A))_i$. Let $h' : \supp\left(U_f\right) \to \{0,1\}^{2n}$ be the following mapping: \[h'(x) = \left\langle x_1,\ldots,x_n \right\rangle \braket{ \left(g\left(B \cup \{a\}\right)\right)_{\mathrm{ord}(a,B \cup \{a\})} \cond\; B \in \binom{S \setminus \{a\}}{k-1} }\]
		Observe that $d(x^a, h'(x^a)) \le \frac{1}{2}$ for every $a \in S$ (because their key parts match) and that if $\kappa(h(x^a)) \ne a$ then $d(x^a,h(x^a)) \ge \frac{1}{6}$, because codewords for different keys are $\frac{1}{3}$-far apart, and the weight of the key is $\frac{1}{2}$.
		
		For every $a \in S$: if $\kappa(h(x^a)) = a$, then $h(x^a) = h'(x^a)$ and so $d(x^a,h(x^a)) = d(x^a,h'(x^a)) \ge \frac{1}{3} d(x^a,h'(x^a))$. Otherwise, $d(x^a,h(x^a)) \ge \frac{1}{6} \ge \frac{1}{3} d(x^a,h'(x^a))$ as well. In total,
		\begin{eqnarray*}
			d(U_f, \mathbf{Par}_k) = d(U_f, h(U_f))
			&=& \sum_{a \in S}{\frac{1}{m} d(x^a, h(x^a))} \\
			&\ge& \frac{1}{3} \sum_{a \in S}{\frac{1}{m} d(x^a, h'(x^a))}
			= \frac{1}{6} d(f, g)
			\ge \frac{1}{6} d(f, \mathbf{par}_k)
		\end{eqnarray*}
	\end{proof}
	
	\begin{lemma}[No useful queries lemma, see Lemma \ref{lemma:no-useful-queries}]
		\label{lemma:no-useful-queries-k}
		Let $f : \binom{S}{k} \to \{0,1\}^k$ be a function, and let $\mathcal{A}$ be a strong $k'$-memory bounded probabilistic algorithm that uses $s$ samples and $q$ queries, which we execute for the input $U_f$. With probability at least $1 - \frac{(k - k')sq}{m}$, for every set $A \in \binom{S}{k}$, the algorithm obtains at most $k'$ bits of $f(A)$.
	\end{lemma}
	
	\begin{proof}	
		Let $T_1,\ldots,T_{s-k'+1}$ be random variables such that $T_i$ is the index of the first query to the $(i+k')$th sample (where $T_{s-k'+1} = q+1$ for convenience). It is exactly when the algorithm must drop one of its old samples. Observe that the $T_i$s split the algorithm execution into $s-k'+1$ phases such that in every individual phase, exactly $k'$ samples are fully accessible (and all the others are not). Let $q_1,\ldots,q_{s-k'+1}$ be the number of queries in each phase.
		
		We proceed by induction. Consider the $i$th phase, and assume that by its end, for every $A \in \binom{S}{k}$, the algorithm obtained at most $k'$ elements of $f(A)$ using queries (observe that this is always the case in the $1$st phase). By the end of the $i$th phase but before the $(i+1)$st one, the algorithm queries at most $q_i$ bits of $f$-values. Every queried point is some bit of $f(A)$ (for $A \in \binom{S}{k}$), and thus the total number of keys that are involved even in one query is at most $k$ (rather than $k'$, because $\card{A} = k$). The probability to draw any new sample with one of these keys that has not been seen before, regarding $f(A)$ with exactly $k'$ known bits, is at most $\frac{(k-k')sq_i}{m}$ (because we have at most $s$ future samples). If it does not happen, then also by the end of the $(i+1)$st phase, for every $A \in \binom{S}{k}$, the algorithm obtained at most $k'$ bits of $f(A)$. By the union bound, the probability that this ``no useful queries'' situation is preserved through all phases is at least $1 - \frac{(k-k') s \sum_{i=1}^{s} q_i}{m}$, which is $1 - \frac{(k-k')sq}{m}$ as desired. When this is the case, the algorithm does not obtain more than $k'$ bits of $f(A)$, for every $A \in \binom{S}{k}$.
	\end{proof}
	
	\begin{theorem}
		\label{th:lower-bound-pi-k-with-k-prime-mem}
		For every $k \ge 2$, every strong $k-1$-memory $\frac{1}{6k}$-test for $\mathbf{Par}_k$ must use at least $\frac{1}{2}\sqrt{m}$ queries.
	\end{theorem}
	
	\begin{proof}
		Consider the following distributions of inputs:
		\begin{itemize}
			\item $D_\mathrm{yes}$ chooses $f : \binom{S}{k} \to \textsc{even} \cap \{0,1\}^k$ uniformly at random and returns $U_f$.
			\item $D_\mathrm{no}$ chooses $f : \binom{S}{k} \to \textsc{odd} \cap \{0,1\}^k$ uniformly at random and returns $U_f$.
		\end{itemize}
		Observe that $D_\mathrm{yes}$ draws an input in $\mathbf{Par}_k$ with probability $1$. The $f$ that is drawn by $D_\mathrm{no}$ is $\frac{1}{k}$-far from $\mathbf{par}_k$. By Lemma \ref{lemma:uf-p-pi-k-f-pi-k-distance}, an input that is drawn from $D_\mathrm{no}$ is $\frac{1}{6k}$-far from $\mathbf{Par}_k$.
		
		According to Lemma \ref{lemma:no-useful-queries-k}, for every strong $k-1$-memory algorithm, with probability higher than $1 - \frac{sq}{m}$, the distribution of answers to queries is completely uniform (because the uniform distributions over \textsc{even} and \textsc{odd} are both $k-1$-uniform), regardless of whether the input is drawn from $D_\mathrm{yes}$ or from $D_\mathrm{no}$. That is, the total variation distance of answers is at most $\frac{sq}{m}$. Without loss of generality $s \le q$, and thus every algorithm that uses less than $\frac{1}{2} \sqrt{m}$ queries cannot be a $\frac{1}{6k}$-test of $\mathbf{Par}_k$.
	\end{proof}
	
	\subsection{Logarithmic, weak $k$-memory $\eps$-test for the parity property}
	
	The $\eps$-testing algorithm for $\mathbf{Par}_k$ is a straightforward generalization of the $\eps$-test for $\mathbf{Sym}$. It makes $O(\eps^{-k}k)$ iterations, each consisting of drawing $k$ samples and validating them.
	
	\begin{algorithm}
		\caption{One-sided $\eps$-test for $\mathbf{Par}_k$, weak $k$-memory, $O(\eps^{-k} k \log n)$ queries}
		\label{alg:par-k}
		\begin{algorithmic}
			\State \textbf{let} $m$ be such that $m = \binom{m-1}{k-1} = n$.
			\For{$\ceil{4\eps^{-k}k}$ \textbf{times}}
			\State \textbf{take} $k$ new samples $x^1, \ldots, x^k$.
			\ForRange{$t$}{$1$}{$k$}
			\State \textbf{query} $x^t_1, \ldots, x^t_{\ceil{\log_2 m}}$, giving $\kappa(x^t)$ as $a^t$.
			\State \textbf{choose} $i \in [m]$, uniformly at random.
			\State \textbf{query} $x^t$ at $i$, giving $x^t_i$.
			\If {$x^t_i \ne (C(a^t))_i$}
			\State \Return \reject \Comment{reject by key invalidity}
			\EndIf
			\EndFor
			\If {$\card{\set{a^1,\ldots,a^k}} = k$}
			\ForRange{$t$}{$1$}{$k$}
			\State \textbf{query} $\Phi_{x^t}(\{a^1,\ldots,a^k\})$, giving $s^t$.
			\EndFor
			\If {$\bigoplus_{i=1}^k s^t = 1$}
			\State \Return \reject \Comment{reject by parity-invalidity}
			\EndIf
			\EndIf
			\EndFor
			\State \Return \accept
		\end{algorithmic}
	\end{algorithm}
	To be able to prove the correctness of the algorithm we need additional definitions. Our goal is to bound the distance $d(P, \mathbf{Par}_k)$ using the probability to reject.
	
	\begin{definition}[Key invalidity, counterpart to Definition \ref{def:Sym:key-invalidity}]
		For an input distribution $P$, we define its \emph{key invalidity} as:
		\[K_{k}(P) = \E_{x \sim P}\left[ d(\braket{x_1,\ldots,x_n},C(\kappa(x))) \right] = \E_{x \sim P, i \sim [n]}\left[x_i \ne (C(\kappa(x)))_i\right]\]
	\end{definition}
	Key invalidity is a measure for ``how far is $P$ from having valid keys'', and it is also the probability of a single iteration of Algorithm \ref{alg:par-k} to reject by key invalidity of $x^1$.
	
	\begin{definition}[parity-invalidity, counterpart to Definition \ref{def:Sym:asymmetry}]
		For an input distribution $P$, we define its \emph{parity-invalidity} as:
		\[I_k(P) = \Pr_{x_1,\ldots,x_k \sim P}\left[ |A| = k \wedge \bigoplus_{i=1}^k \Phi_{x_i}(A) = 1 \ \mathrm{for}\ A = \{\kappa(x_1),\ldots,\kappa(x_k)\} \right]\]
	\end{definition}
	Parity-invalidity is a measure for ``how far is $P$ from being parity-valid'', and it is also the probability of a single iteration of Algorithm \ref{alg:par-k} to reject by parity-invalidity.
	
	\begin{theorem} \label{th:p-pi-k-emd-ubnd}
		If $m > 2k^2$, then for every $\delta > 0$, $d(P,\mathbf{Par}_k) \le \delta + \frac{1}{2}K_k(P) + 2 \delta^{1-k} I_k(P)$.
	\end{theorem}
	
	Before we prove Theorem \ref{th:p-pi-k-emd-ubnd}, we use it to show the correctness of the algorithm.
	
	\begin{theorem}
		Algorithm \ref{alg:par-k} is a one-sided $\eps$-test for $\mathbf{Par}_k$ that uses $O(\eps^{-k} k \log n)$ queries.
	\end{theorem}
	
	\begin{proof}
		Each iteration of the algorithm takes $k$ samples and makes at most $\ceil{\log_2 m} + 2$ queries per sample. For $m > 2k^2$: $\log_2 m\le \frac{1}{k-1} \log_2 n + \log_2 k + 1$ (see Observation \ref{apx:bar2}), hence there are at most $(1 + \frac{1}{k-1})\log_2 n + k \log_2 k + 4k$ queries per iteration. Note that $k \log_2 k \le \log_2 n$ (see Observation \ref{apx:bar3}), hence the number of queries per iteration is bounded by $(2 + \frac{1}{k-1})\log_2 n + 4k$, which is at most $7 \log_2$ for $k \ge 2$. There are $\ceil{4k\eps^{-2}}$ iterations, hence the total number of queries is $O(\eps^{-2} k \log n)$.
		
		Perfect completeness is trivial.
		
		For soundness, consider an $\eps$-far input distribution $P$, and let $\delta=\frac{1}{2^{1/(k-1)}}\eps$. By Theorem \ref{th:p-pi-k-emd-ubnd}, $\eps < \delta + \frac{1}{2}K(P) + 2 \delta^{1-k}I(P)$. Considering the bound $1 - 2^{-1/(k-1)} > \frac{1}{2k}$ and doing the math:
		\begin{eqnarray*}
			\frac{1}{2k}\eps < (1 - \frac{1}{2^{1/(k-1)}})\eps = \eps - \delta < \frac{1}{2}K(P) + 2\delta^{1-k}I(P) = \frac{1}{2}K(P) + \eps^{1-k} I(P).
		\end{eqnarray*}
		This implies that either $K(P) > \frac{1}{2k}\eps$ or $I(P) > \frac{1}{4k}\eps^k$. Either way, the probability to reject $P$ in a single iteration is at least $\frac{1}{4k}\eps^k$, and the probability to do that after $\ceil{4k \eps^{-k}}$ iterations is greater than $\frac{1}{2}$.
	\end{proof}
	
	In preparation to the proof of Theorem \ref{th:p-pi-k-emd-ubnd}, we introduce more definitions and lemmas resembling those for the proof of Theorem \ref{th:alg-psym-correct}. In the following we assume that $P$ is the input distribution over $\{0,1\}^{2n}$.
	
	\begin{definition}[$p_{A,a}$, counterpart to Definition \ref{def:Sym:pab}]
		For $a\in A\in\binom{S}{k}$ all of whose keys appear in the support of $P$, $p_{A,a}$ is defined as $\Pr\left[\Phi_x(A) = 0 \cond \kappa(x) = a \right]$.
	\end{definition}
	
	\begin{definition}[Specific fixing cost, $c_{A,s}$, counterpart to Definition \ref{def:Sym:specific-fixing-cost-abx}]
		For $A=\{a_1,\ldots,a_k\}\in\binom{S}{k}$ where $a_1 < \cdots < a_k$ for which $\Pr_P[a_i] > 0$ for every $1 \le i \le k$, and for a string $s \in \{0,1\}^k$, let the \emph{$s$-fixing cost of $A$} be $c_{A,s} = \frac{1}{2n} \sum_{i=1}^k {(s_i \cdot p_{A,a_i} + (1-s_i)(1 - p_{A,a_i})) \Pr_P[a_i]}$. In other words, $c_A$ is the earth mover's cost of making $P$ consistent at $A$, where its value is $s$.
	\end{definition}
	
	\begin{definition}[Fixing cost, $c_A$, counterpart to Definition \ref{def:Sym:fixing-cost-ab}]
		Let $P$ be a distribution over the set $\{0,1\}^{2n}$. For $A \in \binom{S}{k}$ for which $\Pr_P[a] > 0$ for every $a \in A$, let the \emph{fixing cost of $A$} is defined as $c_{A} = \min_{s \in \textsc{even}} c_{A,s}$. In other words, $c_A$ is the earth mover's cost of making $P$ parity-valid at $A$.
	\end{definition}
	
	\begin{lemma}[A technical bound]
		\label{lemma:even-deviation-is-popular}
		Let $X_1,\ldots,X_k$ be independent random variables such that for every $1 \le i \le k$, $\Pr[X_i=1] \le \frac{1}{2}$. Then $\max_{1 \le i \le k} \Pr[X_i=1] \le \Pr\left[\bigoplus_{i=1}^k X_i = 1\right] \le \frac{1}{2}$.
	\end{lemma}
	\begin{proof}
		Without loss of generality we assume that $\mathrm{Pr}[X_k=1]=\max_{1\leq i\leq k}\mathrm{Pr}[X_i=1]$. For every $1 \le i \le k$ let $p_i = \Pr[X_i=1]$ and for every $1 \le t \le k$ let $r_t = \Pr\left[\bigoplus_{i=1}^t X_i = 1\right]$. 
		
		For the lower bound, observe that:
		\[r_k = (1 - p_k)r_{k-1} + p_k (1 - r_{k-1}) \ge \min\{p_k,1-p_k\} = p_k = \Pr[X_k=1] \]
		
		We prove the upper bound by induction. Trivially, $r_1 = X_1 \le \frac{1}{2}$. For $2 \le t \le k$,
		\[r_t = (1-p_t)r_{t-1} + p_t(1 - r_{t-2}) \underset{(*)}\le \frac{1}{2}(r_{t-2} + 1-r_{t-2}) = \frac{1}{2} \]
		The starred transition is correct because $\alpha(x) = (1-x)r_{t-1} + x (1 - r_{t-1})$ is a non-negative linear mapping (since $r_{t-1} \le \frac{1}{2}$), hence it is non-decreasing in $x$.
	\end{proof}
	
	\begin{lemma}[see Lemma \ref{lemma:Sym-fixing-cost-bound}]
		\label{lemma:p-pi-k-fixing-cost-bound}
		For $A = \{a_1,\ldots,a_k\} \in \binom{S}{k}$ for which $\Pr_P[a] > 0$ for every $a \in A$,
		\[
		c_A
		\le
		\frac{2 \sum_{a \in A} \Pr_P[a]}
		{2 k! n \prod_{a \in A} \Pr_P[a]}
		\Pr_{x_1,\ldots,x_k \sim P} \left[
		\exists \sigma : \left(
		\left(\bigwedge_{i=1}^k \left(\kappa(x_{\sigma(i)}) = a_i\right)\right) \wedge \bigoplus_{i=1}^k \Phi_{x_i}(A) = 1
		\right)
		\right]
		\]
	\end{lemma}
	\begin{proof}
		Let $\rho$ be the probability to find parity-invalidity, conditioned on obtaining the $k$ keys of $A$. By definition,
		\begin{eqnarray*}
			\rho
			&=& \Pr_{x_1,\ldots,x_k \sim P} \left[\bigoplus_{i=1}^k \Phi_{x_i}(A) = 1 \cond \exists \sigma : \bigwedge_{i=1}^k \left(\kappa(x_{\sigma(i)}) = a_i\right) \right] \\
			&=& \Pr_{x_1,\ldots,x_k\sim P}\left[\bigoplus_{i=1}^k \Phi_{x_i}(A) = 1 \cond \bigwedge_{i=1}^k \left(\kappa(x_i) = a_i\right) \right]
		\end{eqnarray*}
		
		Hence
		\begin{eqnarray*}
			\Pr_{x_1,\ldots,x_k \sim P} \left[
			\exists \sigma : \left(
			\left(\bigwedge_{i=1}^k \left(\kappa(x_{\sigma(i)}) = a_i\right)\right) \wedge \bigoplus_{i=1}^k \Phi_{x_i}(A) = 1
			\right)
			\right] = \rho k! \prod_{a\in A} \Pr[a]
		\end{eqnarray*}
		Below we show that $c_A \le \frac{2\rho}{2n}\sum_{i=1}^k \Pr_P[a_i]$, which completes the proof.
		
		Let $s = \arg \min_{s \in \textsc{even}} c_{A,s}$, and let $m$ be a majority string, that is, a string such that  for every $1 \le i \le k$, $\Pr\left[\Phi_x(A) \ne m_i \cond \kappa(x) = a_i\right] \le \frac{1}{2}$. Note that $m$ is not necessarily unique, but every arbitrary choice would fit for the analysis.
		
		For every $1 \le i \le k$, let $p_i = \Pr\left[\Phi_x(A) \ne s_i \cond \kappa(x) = a_i\right]$ be the probability of the $i$th bit to deviate from $s$ and let $q_i = \Pr\left[\Phi_x(A) \ne m_i \cond \kappa(x) = a_i\right]$ be the probability to deviate from the majority. In the following cases we use the fact that every two words that differ by one bit have different parities.
		
		\paragraph{Case I. $m$ has odd parity} By Lemma \ref{lemma:even-deviation-is-popular}, the probability to draw a string that has the same parity as $m$, which is odd, is at least $\frac{1}{2}$. Hence $c_A \le \frac{1}{2n} \sum_{i=1}^k \Pr_P[a_i] p_i \le
		\frac{2\rho}{2n}\sum_{i=1}^k \Pr_P[a_i]$.
		
		\paragraph{Case II. $m$ has even parity} In this case, $m=s$, because of the minimality of $s$. That is, $p_i \le \frac{1}{2}$ for every $1 \le i \le k$. By Lemma \ref{lemma:even-deviation-is-popular}, the probability to draw an odd-parity string is at least $\max_{1 \le i \le k} p_i$. Hence $\rho \ge p_i$ and $c_A = \frac{1}{2n} \sum_{i=1}^k \Pr_P[a_i] p_i \le \frac{1}{2n} \sum_{i=1}^k \Pr_P[a_i]\max p_i \le \frac{\rho}{2n} \sum_{i=1}^k \Pr_P[a_i]$.
	\end{proof}
	
	\begin{lemma}[A technical bound] \label{lemma:techbnd:nasty}
		For every $m > 2k^2$ and for every $A = \{a_1,\ldots,a_k\}$, if we have $\Pr_P[a] \ge \frac{\delta}{m}$ for every $a \in A$, then
		$\frac{2\sum_{a \in A} \Pr_P[a]}{2 k! n\prod_{a \in A} \Pr_P[a]} \le 2\delta^{1-k}$.
	\end{lemma}
	\begin{proof}
		Without loss of generality we assume that $\Pr_P[a_k] \ge \Pr_P[a_1],\ldots,\Pr_P[a_{k-1}]$. Based on the bound $m^{k-1} \le \frac{e^{1/2} (m-1)!}{(m-k)!}$ (for every $m > 2k^2$, see Observation \ref{apx:bar}),
		\[
		\frac{2\sum_{a \in A} \Pr_P[a]}{2 k! n\prod_{a \in A} \Pr_P[a]}
		\le \frac{2k \Pr[a_k]}{ 2 k!n \Pr[a_k] \left(\delta/m\right)^{k-1} }
		= \frac{2\delta^{1-k} \cdot m^{k-1}}{2 (k-1)! n}
		\le \frac{2\delta^{1-k} \cdot e^{1/2} (m-1)!}{2 (k-1)!(m-k)! n}
		\le 2\delta^{1-k}
		\]
	\end{proof}
	
	Now we are able to prove Theorem \ref{th:p-pi-k-emd-ubnd}.
	
	\begin{proof}[Proof (of Theorem \ref{th:p-pi-k-emd-ubnd})]
		Let $\tilde{S}$ be the set of keys whose probability to be drawn is at least $\frac{\delta}{m}$. Let $f : \binom{\tilde{S}}{k} \to \{0,1\}^k$ be defined by $f(A) = \arg \min_{s \in S} c_{A,s}$. Let $h(x)$ be the following map:
		\begin{align*}
			h(x) = C(\kappa(x))\braket{ \begin{cases}
					(f(A\cup\{\kappa(x)\})_{\mathrm{ord}(\kappa(x),A\cup\{\kappa(x)\})} & \kappa(x) \in \tilde{S}, A \subseteq \tilde{S} \\
					\phi_x(A) & \mathrm{otherwise}
				\end{cases} \cond A \in \binom{S \setminus \{\kappa(x)\}}{k - 1}}
		\end{align*}	
		The distribution $h(P)$ does not necessarily belong to $\mathbf{Par}_k$, but it is $\delta$-close to it: if we delete all elements whose key is not in $\tilde{S}$ (and transfer their probabilities to other elements arbitrarily), the result distribution does belong to $\mathbf{Par}_k$.
		
		For every $A \in \binom{S}{k}$, let $\mathcal{B}_A$ be the event of catching parity-invalidity in $A$. Based on this notation and the last bound,
		\begin{eqnarray*}
			& d(P,h(P)) &\le \quad \sum_{a \in S} \Pr_P[a] \E_{x \sim P}\left[d(x,h(x)) \cond \kappa(x) = a\right] \\
			&\le& \hspace{-12pt} \frac{1}{2} K_k(P) + \frac{1}{2} \sum_{a \in \tilde{S}} \Pr_P[a] \E_{x \sim P}\left[d(\braket{x_{n+1},\ldots,x_{2n}},\braket{h_{n+1}(x),\ldots,h_{2n}(x)}) \cond \kappa(x) = a\right] \\
			&\le& \hspace{-12pt} \frac{1}{2}K_k(P) + \sum_{A \in \binom{\tilde{S}}{k}} c_A\\
			&\underset{(*)}\le& \hspace{-12pt} \frac{1}{2}K_k(P) + \sum_{A \in \binom{\tilde{S}}{k}} \frac{2 \sum_{a \in A} \Pr_P[a]}{2 k! n \prod_{a \in A} \Pr_P[a]} \Pr\left[\mathcal{B}_A\right]\\
			&\underset{(**)}\le& \hspace{-12pt} \frac{1}{2}K_k(P) + \sum_{A \in \binom{\tilde{S}}{k}} 2 \delta^{1-k} \Pr\left[\mathcal{B}_A\right]\\
			&=& \hspace{-12pt} \frac{1}{2}K_k(P) + 2 \delta^{1-k} \sum_{A \in \binom{\tilde{S}}{k}} \Pr\left[\mathcal{B}_A\right] \\
			&\le& \hspace{-12pt} \frac{1}{2}K_k(P) + 2 \delta^{k-1} I_k(P) .
		\end{eqnarray*}
		The first transition is correct because we can use a transfer distribution that maps every $x$ to its $h(x)$, and the rightmost sum only considers keys in $\tilde{S}$ because $h$ does not modify values of samples with rare keys. The starred transition is correct because of Lemma \ref{lemma:p-pi-k-fixing-cost-bound}, and the doubly-starred transition is implied by the technical bound Lemma \ref{lemma:techbnd:nasty} since $\Pr[a] \ge \frac{\delta}{m}$ for all $a \in A$ and $m > 2k^2$.
		
		Hence $d(P,\mathbf{Par}_k)
		\le \delta + d(P,h(P))
		\le \delta + \frac{1}{2}K_k(P) + 2 \delta^{1-k} I_k(P)$.
	\end{proof}
	

	\bibliographystyle{alpha}
	\bibliography{refining-the-adaptivity-notion-in-the-huge-object-model}
	\appendix

	\section{Calculations}
	
	\begin{observation}
		\label{apx:bar}
		For $k \ge 2$ and $m > ak^2$, $m^{k-1} < e^{1/a} \frac{(m-1)!}{(m-k)!}$
	\end{observation}
	
	\begin{proof} This follows from the following calculation:
		\[\frac{(m-1)!}{(m-k)!} \ge (m-k)^{k-1} = m^{k-1} \left(1 - \frac{k}{m}\right)^{k-1} \ge m^{k-1} \left(1 - \frac{1}{ak}\right)^{k-1} > e^{-1/a} m^{k-1}\]
	\end{proof}
	
	\begin{observation}
		\label{apx:bar2}
		For $n=\binom{m-1}{k-1}$ and $m \ge 2 k^2$ and $k \ge 2$, $\ceil{\log_2 m} \le \frac{1}{k-1}\log_2 n + \log_2 k + 1$.
	\end{observation}
	
	\begin{proof}
		This is implied from Observation \ref{apx:bar} using $a=2$:
		\begin{eqnarray*}
			m^{k-1} &<& 2(k-1)! n \\
			m &\le& (2(k-1)!)^{1/(k-1)} n^{1/(k-1)} \le k n^{1/(k-1)} \\
			\log_2 m &<& \frac{1}{k-1} \log_2 n + \log_2 k
		\end{eqnarray*}
		The conclusion $\ceil{\log_2 m} < \frac{1}{k-1} \log_2 n + \log_2 k + 1$ is now immediate.
	\end{proof}
	
	\begin{observation}
		\label{apx:bar3}
		For $n = \binom{m-1}{k-1}$ and $m \ge k^2$, $k \log_2 k \le \log_2 n$.
	\end{observation}
	
	\begin{proof}
		Noting that
		\[n = \binom{m-1}{k-1} \ge \binom{2k^2-1}{k-1} \ge \frac{(2k^2-k)^{k-1}}{(k-1)!} \ge \frac{(2k)^{k-1} (k-1)^{k-1}}{(k-1)^{k-1}} = (2k)^{k-1} \ge k^k\]
		It follows that
		\[\log_2 n \ge k \log_2 k\]
	\end{proof}
	
	\begin{observation}
		\label{apx:tbnd1}
		For every $0 \le h \le q < n$, $\left(\lfrac{n}{\left(n-q\right)}\right)^h < e^{q^2/(n-q)}$.
	\end{observation}
	
	\begin{proof}
		Follows from the following:
		\[\left(\frac{n}{n-q}\right)^h \le \left(\frac{n}{n-q}\right)^q = \left(1 + \frac{q}{n-q}\right)^q < e^{\lfrac{q^2}{\left(n - q\right)}}\]
	\end{proof}
	

\end{document}